\documentclass[a4paper,fleqn]{cas-sc}
\usepackage{ifpdf}

\usepackage{graphicx}
\graphicspath{{figures/}}
\usepackage{epsfig} 
\usepackage{epstopdf} 

\usepackage{amsmath}
\usepackage{amssymb}
\usepackage{amsfonts} 
\usepackage{mathrsfs}
\usepackage{amsthm}

\usepackage{array}

\usepackage{float}
\usepackage{subcaption} 
\usepackage{wrapfig}
\usepackage[font=small]{caption}

\usepackage{hyperref}

\usepackage{csquotes} 
\usepackage{siunitx} 
\usepackage{gensymb} 

\theoremstyle{definition}
\newtheorem{theorem}{Theorem}[section]

\newtheorem{proposition}[theorem]{Proposition}
\newtheorem{corollary}[theorem]{Corollary}

\newtheorem{remark}[theorem]{Remark}

\newtheorem{lemma}[theorem]{Lemma}

\newcommand{\half}{\frac{1}{2}}

\newcommand{\bb}[1]{\mathbb{#1}}
\newcommand{\cl}[1]{\mathcal{#1}}

\newcommand{\R}[1]{\bb{R}^{#1}}

\newcommand{\TwoTwoMat}[4]{
\begin{pmatrix}
#1 & #2 \\
#3 & #4
\end{pmatrix}}
\newcommand{\TwoVec}[2]{
\begin{pmatrix}
#1\\
#2 
\end{pmatrix}}

\newcommand{\ThrVec}[3]{
\begin{pmatrix}
#1\\
#2\\
#3
\end{pmatrix}}

\newcommand{\map}[3]{{#1}:{#2}\rightarrow {#3}}

\newcommand{\pair}[2]{\left\langle \left.  #1  \right|  #2 \right\rangle}

\newcommand{\extd}{\textrm{d}}
\newcommand{\dt}{\frac{d}{dt}}

\newcommand{\pH}{port-Hamiltonian }

\newcommand{\parGrad}[2]{\frac{\partial #1}{\partial #2}}
\newcommand{\subTxt}[2]{#1_{\text{\normalfont #2}}}
\newcommand{\twist}[3]{\mathcal{V}_{#1}^{#3,#2}}
\newcommand{\twistV}{\mathcal{V}}
\newcommand{\twistTd}[3]{\tilde{\mathcal{V}}_{#1}^{#3,#2}}
\newcommand{\twistDot}[3]{\dot{\mathcal{V}}_{#1}^{#3,#2}}

\newcommand{\angVel}[3]{\omega_{#1}^{#3,#2}}

\newcommand{\angVelTd}[3]{\tilde{\omega}_{#1}^{#3,#2}}

\newcommand{\linVel}[3]{v_{#1}^{#3,#2}}

\newcommand{\baseVel}{\text{v}}
\newcommand{\baseIner}{I}
\newcommand{\baseMom}{\text{p}}
\newcommand{\baseWrench}{\text{w}}

\newcommand{\spBaseVel}{\mathfrak{g}_b}

\begin{document}
\let\WriteBookmarks\relax
\def\floatpagepagefraction{1}
\def\textpagefraction{.001}
\shorttitle{Port-Hamiltonian VMS}
\shortauthors{R. Rashad}
\title[mode = title]{The Port-Hamiltonian Structure of Vehicle Manipulator Systems}
\tnotemark[1]
\tnotetext[1]{This work has been funded by King Fahd University of Petroleum and Minerals under Project EC251005.}

\author[1]{Ramy Rashad}[type=editor,
                        role=Researcher,
                        orcid=0000-0002-9083-0504]
\cormark[1]
\fnmark[1]
\ead{ramy.rashad@kfupm.edu.sa}
\affiliation[1]{organization={Control and Instrumentation Engineering department and Interdisciplinary Research Center for Smart Mobility and Logistics, King Fahd University of Petroleum and Minerals},
                city={Dhahran},
                postcode={34464}, 
                country={Saudi Arabia}}

\begin{abstract}
    This paper presents a port-Hamiltonian formulation of vehicle-manipulator systems (VMS),
a broad class of robotic systems including aerial manipulators, underwater manipulators, 
space robots, and omnidirectional mobile manipulators. 
Unlike existing Lagrangian formulations that obscure the underlying 
energetic structure, the proposed port-Hamiltonian formulation explicitly reveals the 
energy flow and conservation properties of these complex mechanical systems. 
We derive the port-Hamiltonian dynamics from first principles using Hamiltonian reduction theory.
Two complementary formulations are presented: a standard form that directly exposes 
the energy structure, and an inertially-decoupled form that leverages the principal bundle structure of the VMS configuration space and is 
particularly suitable for control design and numerical simulation. 
The coordinate-free geometric approach we follow avoids singularities associated with local parameterizations of the base orientation.
We rigorously establish  the mathematical equivalence between our port-Hamiltonian formulations and existing reduced Euler-Lagrange 
and Boltzmann-Hamel equations found in the robotics and geometric mechanics literature. 

\end{abstract}


\maketitle

\section{Introduction}\label{sec:intro}
Vehicle manipulator systems (VMS) are a broad class of robotic systems that consist of a manipulator mounted on a mobile base.
Examples of VMS include aerial manipulators~\cite{Nava2020}, underwater manipulators~\cite{Simetti2021}, 
space robots~\cite{giordano2021compliant}, and omnidirectional mobile manipulators~\cite{kim2019whole}.

The equations of motion of VMS have been extensively studied in the robotics and geometric mechanics literature and are indespensable for 
analysis, model-based control, parameter identification, and numerical simulation.
In the robotics community, the equations of motion of a VMS are typically derived by the recursive Newton-Euler approach starting from the governing equations 
of a single rigid body \cite{Huang2017}.
In the geometric mechanics community, the equations of motion are usually derived using variational principles on the tangent bundle of the configuration space.


The equations of motion of a VMS possess several hidden structures that can be exploited for different robotics applications.
On one hand, there is the kinematic tree structure of manipulators that can be exploited for computational efficiency in calculating the mass and Coriolis-centrifugal
matrices as well as their derivatives needed for model-based control and optimization \cite{Garofalo2013,Echeandia2021}.
Another aspect is the geometric structure given by the non-Euclidean nature of the configuration space which is fundamental to modern treatments of rigid-body dynamics and robotics 
\cite{lynch2017modern,murray2017mathematical}.
In particular, the configuration space of a VMS has a principal bundle structure which has been exploited for analysis \cite{Saccon2017},
 control \cite{giordano2021compliant} and locomotion planning \cite{hatton2015nonconservativity}.

There is one structure though that has not been fully exploited in the robotics literature, namely the energetic structure of the equations of motion.
Most existing works formulate the equations of motion using the Lagrangian approach which usually obscures such underlying structure.
The energetic structure is usually represented using symplectic or Poisson structures in the Hamiltonian framework, \textit{cf.} \cite{Chhabra2015}.
However, the Hamiltonian approach is not widely used in robotics due to its inapplicability to control synthesis 
as most Hamiltonian systems represent isolated systems that preserve energy.
On the other hand, the \pH framework \cite{van2014port,duindam2009modeling_b} extends the Hamiltonian approach (using port-based modeling) to 
open dynamical systems that can exchange energy with their environment through power ports.

The \pH framework has the ability to highlight the energetic and geometric structure of complex mechanical systems (e.g., \cite{Rashad2025,RafeeNekoo2025, PONCE2024434, WARSEWA20211528}), which can then be exploited for new insights in analysis, simulation and control (e.g., \cite{PONCE2026116403,DONG2025116055, DEJONG2026116775}).
In the robotics and control communities, there have been many successful stories of the insight that energy-based thinking has brought to the field,
such as impedance and admittance control \cite{Hogan1985ImpedanceControl}, energy-shaping control \cite{ortega2002putting},
virtual energy tanks \cite{califano2022use}, control by interconnection \cite{ortega2008control}, and learning robot dynamics \cite{Duong2024}.

The main goal of this paper is to present a \pH formulation of VMS dynamics that reveals the underlying energetic structure of these systems.
To the best of the authors' knowledge, this is the first work that systematically derives the \pH dynamics of VMS from first principles using Hamiltonian reduction theory.
We present two complementary \pH formulations of the VMS dynamics: one form in the standard velocity variables and another inertially-decoupled form in the principal bundle variables.

Our work relies on the geometric formulation of robotic systems using Lie group theory pioneered by Brockett \cite{Brockett1984} and evolved 
over the years by many researchers \cite{lynch2017modern,murray2017mathematical}.
In particular, our work was inspired by the recent works of Mishra et al. \cite{Mishra2023} and Moghaddam and Chhabra \cite{Moghaddam2024} 
that derived the Lagrangian equations of motion on the principal bundle configuration space.
In the \pH literature, Duindam and Stramigioli \cite{Duindam2008} presented a \pH formulation of generic open-chain mechanisms, with generic 
holonomic and nonholonomic joints, using the Boltzmann-Hamel equations.
Their work considered generic quasi-velocities and neither considered the principal bundle structure of VMS nor a systematic derivation of the \pH dynamics
 from first principles.

The main contributions of this work can be explicitly summarized as follows:
\begin{enumerate}
    \item Formulate the VMS dynamics in the \pH framework.
    \item Derive the \pH form from first principles using Hamiltonian reduction and not by manipulating the Lagrangian equations of motion.
    \item Present an inertially-decoupled \pH formulation that leverages the principal bundle structure of the configuration space.
    \item Show the equivalence between the proposed dynamic models and Lagrangian formulations in the literature.
    \item Consider symmetry-breaking generalized forces (such as gravity, actuator torques, and end effector wrenches) which are usually
     neglected in geometric mechanics treatments but essential for robotics applications.
\end{enumerate}

The remainder of this paper is organized as follows.
Section~\ref{sec:math_prelim} introduces preliminaries and establishes the notation used throughout the paper.
We derive the \pH formulations of a fixed-base manipulator and moving-base separately in Sections ~\ref{sec:manip_dyn} and \ref{sec:base_dyn}, respectively, as 
intermediate steps towards the full VMS dynamics.
In Section~\ref{sec:fm_dyn}, we present our two proposed formulations for VMS dynamics.
Section~\ref{sec:fm_dyn_lagrangian} demonstrates the equivalence between our proposed formulations and existing ones in the literature.
Section~\ref{sec:discussion} discusses the implications and potential applications of the proposed framework.
Finally, we conclude in Section~\ref{sec:conclusion}.

\section{Mathematical Preliminaries} \label{sec:math_prelim}
\subsection{Rigid Body Kinematics and Dynamics on $SE(3)$}

Let $\{a\}$ denote an orthonormal frame with origin at point $o_a$ and axes aligned with unit vectors $\{\hat{x}_a, \hat{y}_a, \hat{z}_a\}$.
The relative pose of frame $\{b\}$ with respect to frame $\{a\}$ is represented by the pair $(R_b^a, \xi_b^a) \in SE(3)$, where $R_b^a \in SO(3)$ 
denotes the relative orientation of the two frames and $\xi_b^a \in \mathbb{R}^3$ denotes the position of the origin $o_b$ expressed in frame $\{a\}$.
We shall use the homogeneous transformation matrix to represent the relative pose of frame $\{b\}$ with respect to frame $\{a\}$ as
\begin{equation}
    H_b^a = \TwoTwoMat{R_b^a}{\xi_b^a}{0_{1\times 3}}{1} \in SE(3).    
\end{equation}
With an abuse of notation, we shall also denote the space of homogeneous transformation matrices as $SE(3)$, the special Euclidean group in three dimensions.

The relative twist (generalized velocity) between any two frames $\{a\}$ and $\{b\}$ is represented by
\begin{equation}
    \twistTd{a}{b}{a} = \TwoTwoMat{\angVelTd{a}{b}{a}}{\linVel{a}{b}{a}}{0_{1\times3}}{0} :=  H_b^a\dot{H}_a^b  \in se(3),
\end{equation}
\begin{equation}
    \twistTd{a}{b}{b} = \TwoTwoMat{\angVelTd{a}{b}{b}}{\linVel{a}{b}{b}}{0_{1\times3}}{0} :=  \dot{H}_a^b H_b^a  \in se(3),
\end{equation}
where both twists represent the same physical quantity, but expressed in different frames, while $se(3)$ denotes the Lie algebra associated with $SE(3)$.
We denote by $\linVel{a}{b}{\star} \in \R{3}$ the linear velocity component of the twist and by $\angVelTd{a}{b}{\star} \in so(3)$ the skew-symmetric matrix representing 
the angular velocity component of the twist, where $\star$ indicates the frame in which the twist is expressed.

We can associate to any $\tilde{\omega}\in so(3)$ the vector $\omega \in \bb{R}^3 $ defined such that
$\forall x \in \R{3},\ \tilde{\omega} x = \omega\times x,$ i.e., the cross product of $\omega$ with $x$.
Consequently, a twist $\twistTd{a}{b}{\star}$ can then be represented as a 6-dimensional vector as
\begin{equation}\label{eq:twist_vector}    
    \twist{a}{b}{\star} = \TwoVec{\angVel{a}{b}{\star}}{\linVel{a}{b}{\star}}\in \R{6}.
\end{equation}
With an abuse of notation, we shall refer to both $\twistTd{a}{b}{\star} \in se(3)$ and $\twist{a}{b}{\star} \in \R{6}$ as the twist,
representing the relative generalized velocity between frames $\{a\}$ and $\{b\}$.

Consider a spatial (inertial) reference frame, denoted by $\{s\}$, and a body-fixed frame attached to a rigid body, denoted by $\{b\}$.
We shall refer to the relative twist $\twist{b}{s}{b}$ as the body twist of that rigid body.
The equations of motion of the rigid body are given by the Euler-Poincar\'e equations \cite{bullo2019geometric} as
\begin{equation} \label{eq:rigid_body_dynamics}
    \cl{I}^b_b \twistDot{b}{s}{b} = ad_{\twist{b}{s}{b}}^\top \cl{I}^b_b \twist{b}{s}{b} + \cl{W}^b,
\end{equation}
where $\cl{I}^\star_b \in \R{6\times6}$ denotes the generalized inertia matrix of the rigid body expressed in frame $\{\star\}$,
$\cl{W}^\star \in (\R{6})^*$ denotes the external wrench acting on the rigid body expressed in frame $\{\star\}$, and $ad_{\twistV}^\top \in \R{6\times6}$ denotes the 
dual to the adjoint operator $ad_{\twistV} \in \R{6\times6}$ associated with the Lie algebra $se(3)$,
defined as
\begin{equation}
    ad_{\twistV} := \TwoTwoMat{\tilde{\omega}}{0_{3\times 3}}{\tilde{v}}{\tilde{\omega}}, \qquad \twistV = \TwoVec{\omega}{v} \in \R{6}. 
\end{equation}


Twists and wrenches transform between different frames according to the adjoint and co-adjoint maps of the Lie group $SE(3)$, respectively, given by
\begin{equation}
    \twist{b}{s}{a} = Ad_{H_b^a} \twist{b}{s}{b}, \qquad \cl{W}^a = Ad_{H_b^a}^{-\top} \cl{W}^b,
\end{equation}
where
\begin{equation}
    Ad_{H_b^a} = \TwoTwoMat{R_b^a}{0_{3\times 3}}{\tilde{\xi}_b^a R_b^a}{R_b^a},
\end{equation}
which highlights how twists and wrenches belong to different spaces. We emphasize this by indicating the space of wrenches as $(\R{6})^*$.
It is worth noting that \eqref{eq:rigid_body_dynamics} is invariant under coordinate changes of the body-fixed frame \cite{lynch2017modern,Hong2022}.

\subsection{Joint Kinematics and Dynamics on Lie Subgroups of $SE(3)$}
Subgroups of $SE(3)$ represent constrained displacements of a rigid body that arise due to the presence of joints.
We introduce in this section the minimum geometric background required to describe kinematics and dynamics of relevant joints that arise in common VMS.
The reader is referred to \cite{duindam2009modeling_a,chhabra2014unified} for a detailed treatment of this geometric approach for global parametrization of joints.

Consider a joint connecting two rigid bodies, denoted by body $i$ and body $j$, respectively.
The displacement subgroup of $SE(3)$ associated with the joint connecting body $i$ to body $j$ is denoted by $G_i^j \subseteq SE(3)$, and 
represents the set of all possible relative poses of body $i$ with respect to body $j$ that are allowed by the joint.
We have that $G_i^j$ is a $b$-dimensional Lie subgroup of $SE(3)$, where $b \leq 6$ denotes the number of degrees of freedom (DoF) of the joint.
We denote by $q_i \in G_i$ the joint configuration of joint $i$ connecting body $i$ to body $j$, where $G_i$ denotes the configuration manifold of joint $i$ and is 
a $b$-dimensional Lie group isomorphic to $G_i^j$.
The isomophism between $G_i$ and $G_i^j$ is denoted by $\map{\varphi_i}{G_i}{G_i^j}$ such that the relative pose of body $i$ with respect to body $j$ is given by
\begin{equation}
    H_i^j = \varphi_i(q_i) \in G_i^j.
\end{equation}

Let $\mathfrak{g}_i^j\subseteq se(3)$ and $\mathfrak{g}_i$ denote the Lie algebras associated with the Lie groups $G_i^j$ and $G_i$, respectively.
An element in $\mathfrak{g}_i$ can be identified with the joint velocity vector of joint $i$, denoted by $\baseVel_i \in \bb{R}^b$, while 
an element in $\mathfrak{g}_i^j$ can be identified with the relative twist $\twist{i}{j}{i} \in \R{6}$ allowed by the joint.
For simplicity of notation, we shall simply say that $\baseVel_i \in \mathfrak{g}_i$ and $\twist{i}{j}{i}\in \mathfrak{g}_i^j$ from now on.
We have that $\baseVel_i $ is related to the time derivative of the joint configuration $q_i$ by the map $\map{\chi_{q_i}}{\mathfrak{g}_i}{T_{q_i} G_i}$, such that
\begin{equation}\label{eq:chi_map_def}
    \dot{q}_i = \chi_{q_i}(\baseVel_i) \in T_{q_i} G_i.
\end{equation}
The induced Lie algebra isomorphism that relates the joint velocities $\text{v}_i$ to the relative twist $\twist{i}{j}{i}$ of body $i$ with respect to body $j$ is denoted by
$\map{\iota_i}{\mathfrak{g}_i}{\mathfrak{g}_i^j}$ such that
\begin{equation}
\twist{i}{j}{i} = \iota_i(\baseVel_i) = \cl{S}_i^{i,j} \baseVel_i  \in \mathfrak{g}_i^j,
\end{equation}
where $\iota_i$ is a linear map that can be represented by the configuration-independent matrix $ \cl{S}_i^{i,j} \in \bb{R}^{6 \times b}$.
The dual map $\map{\iota_i^*}{(\mathfrak{g}_i^j)^*}{\mathfrak{g}_i^*}$ relates the wrenches $\cl{W}^i \in (\mathfrak{g}_i^j)^* \cong (\R{6})^*$ 
acting on body $i$ to the joint's generalized forces $\baseWrench_i \in \mathfrak{g}_i^*\cong (\R{b})^*$ such that
$$(\cl{W}^i)^\top \twist{i}{j}{i} = (\baseWrench_i)^\top \baseVel_i.$$

For the reader's convenience, we present in Appendix \ref{appendix:joint_subgroups} the kinematics and dynamics of common joints used in VMS systems, 
including 1-DoF revolute and prismatic joints, 3-DoF planar, and 6-DoF floating joints.

\subsection{Standard Formulation of VMS Dynamics}
We consider a VMS composed of a robotic moving base and a serial kinematic chain consisting of $n$ rigid links connected by $n$ actuated 1 DoF joints.
We consider the case where the base is free to translate and rotate without any kinematic constraints.
The base movement is modeled either with a 6-DoF floating or a 3-DoF planar (virtual) joint.
This formulation allows for modeling a wide range of VMS, including aerial, underwater, space, or omnidirectional ground bases.
For ease of presentation, we consider only single-arm VMS, however, 
the presented formulation is extendible to multi-arm manipulators.

\begin{figure}
    \centering
    \includegraphics[width=0.99\columnwidth]{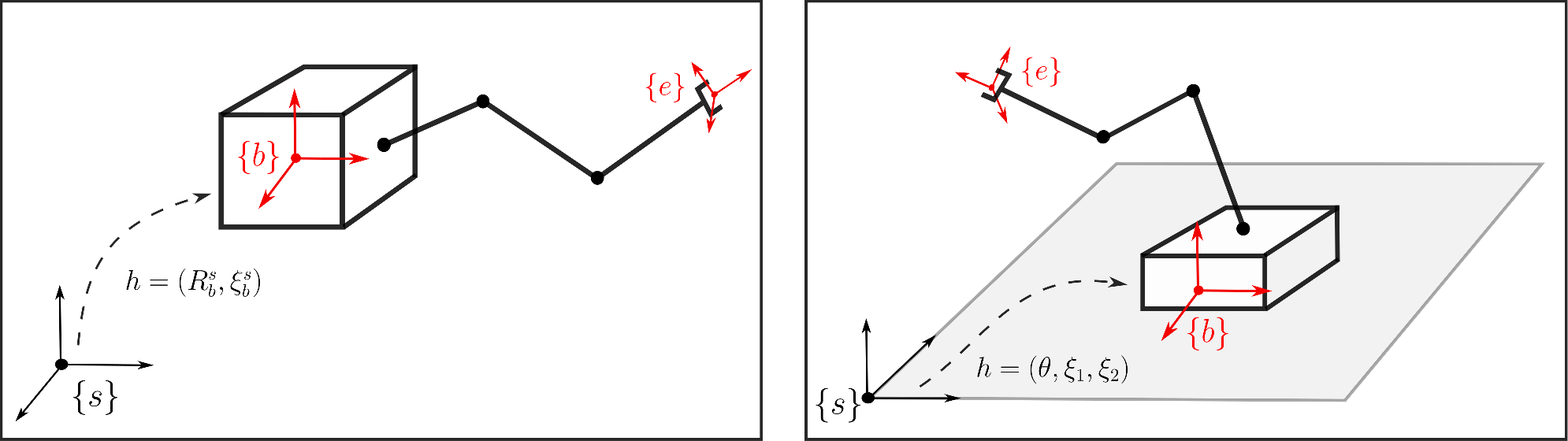}
    \caption{Illustration of vehicle-manipulator systems. Left figure shows a floating-base manipulator $h\in G_b =  SE(3)$ while the right figure shows a ground mobile manipulator $h\in G_b = \bb{S}^1\times \R{2}$.}
    \label{fig:floating_base_manipulator}
\end{figure}

We denote by $\{b\}$ the body-fixed frame attached to the moving base and by $\{i\}$ the body-fixed frame attached to link $i$ of the manipulator, for $i=1,\ldots,n$.
We denote by $\{s\}$ and $\{e\}$ the spatial (inertial) and end effector reference frames, respectively, as illustrated in Fig. \ref{fig:floating_base_manipulator}.
The configuration of the manipulator subsystem is described by the pair $q:=(q_1,\cdots,q_n) \in Q_m$, such that $q_i \in G_i$ denotes the joint configuration of 
joint $i$ connecting link $i$ to its parent link $i-1$, with $\{0\} = \{b\}$. We denote the manipulator configuration space by $Q_m := G_1 \times G_2 \times \cdots \times G_n$.
The configuration of the moving base is described by $h \in G_b$, such that the pose of the base frame $\{b\}$ with respect to the spatial frame $\{s\}$ is given by $H_b^s= \varphi_b(h)$.
The overall configuration of the VMS is then given by the pair $(h,q) \in Q:= G_b \times Q_m$.

By a left trivialization of the tangent bundle $TG_b$, the generalized velocity of the VMS can be represented by the pair
$(\baseVel, \dot{q}) \in \mathfrak{g}_b \times T_q Q_m$, such that $\dot{h} = \chi_h(\baseVel) \in T_h G_b$, where $\mathfrak{g}_b$ denotes the Lie algebra 
associated with the Lie group $G_b$,
We denote by $\baseVel \in \mathfrak{g}_b$ the twist of the moving base and by $\dot{q} \in T_q Q_m$ the joint velocities of the manipulator.
Note that we have that $T_qQ_m \cong \R{n}$ and $\spBaseVel\cong \R{b}$, where $b=3$ or $6$ depending on whether the base's motion is planar or floating, respectively.

\subsubsection{Forward Kinematics}
The forward kinematics of all bodies in the VMS can be computed recursively using the following equations:
\begin{align*}
    H_b^s = \varphi_b(h),\qquad \qquad
    H_i^s = H_b^s H_i^b, \qquad \qquad
    H_i^b = H_{i-1}^b H_i^{i-1}, \qquad i=1,\ldots,n,
\end{align*}
where $H_i^{i-1} = \varphi_i(q_i) \in G_i^j$ denotes the relative pose of link $i$ with respect to its parent link $i-1$ allowed by joint $i$.
Note that in general $H_i^b$ depends on $q$ and can be computed using the product of exponentials formula \cite{lynch2017modern}.

\subsubsection{Differential Kinematics}
The body twist of the moving base is given by
\begin{equation}\label{eq:base_twist}
    \twist{b}{s}{b} = \cl{S}_b^{b,s} \baseVel,
\end{equation}
with $ \cl{S}_b^{b,s} \in \R{6\times b}$ denoting the matrix representation of the mapping $\iota_b$.
The body twist of each link $\{i\}$ can be expressed as
\begin{equation}\label{eq:link_twist}
    \twist{i}{s}{i} = J_i^{i,s} (q) \TwoVec{\baseVel}{\dot{q}},
\end{equation}
where $\map{J_i^{i,s}(q)}{\spBaseVel\times T_q Q_m}{\mathfrak{g}_i^j}$ denotes the geometric Jacobian of link $i$ with respect to the spatial frame $\{s\}$, expressed in frame $\{i\}$, which 
can be expressed as
\begin{equation}\label{eq:jacobian_partition}
    J_i^{i,s}(q) =  \left[ \text{Ad}_{H_b^i(q)}\cl{S}_b^{b,s} , J_i^{i,b}(q)\right],
\end{equation}
with $\map{J_i^{i,b}(q)}{T_q Q_m}{\mathfrak{g}_i^j}$ denoting the geometric Jacobian of link $i$ with respect to the base frame $\{b\}$, expressed in frame $\{i\}$, 
which can be computed as
\begin{equation*}
    J_i^{i,b}(q) =  \left[ \text{Ad}_{H_1^i(q)}\cl{S}_1^{1,0} , \text{Ad}_{H_2^i(q)}\cl{S}_2^{2,1}, \cdots, \cl{S}_i^{i,i-1},  0_{6\times (n-i)}\right],
\end{equation*}
with $ \cl{S}_i^{i,i-1} \in \R{6}$ denoting the matrix representation of the mapping $\iota_i$, for $i=1,\ldots,n$.


\subsubsection{Equations of Motion}
The total kinetic energy of the VMS is given by the sum of the kinetic energies of the moving base and each link as
\begin{equation}
    \subTxt{E}{kin} = \frac{1}{2} (\twist{b}{s}{b} )^\top \cl{I}_b^b \twist{b}{s}{b} + \frac{1}{2} \sum_{i=1}^n (\twist{ i}{s}{i})^\top \cl{I}_i^i \twist{ i}{s}{i},
\end{equation}
where $\cl{I}_b^b \in \R{b\times b}$ denotes the generalized inertia matrix of the moving base expressed in frame $\{b\}$, while $\cl{I}_i^i \in \R{6\times 6}$ denotes the generalized inertia matrix of link $i$ expressed in frame $\{i\}$.
Using \eqref{eq:base_twist} and \eqref{eq:link_twist} and the partioning in \eqref{eq:jacobian_partition}, the kinetic energy can be expressed in the 
standard quadratic form as \cite{Mishra2023,Moghaddam2024}
\begin{equation}\label{eq:fm_kinetic_energy}
    \subTxt{E}{kin} = \frac{1}{2} \TwoVec{\baseVel}{\dot{q}}^\top \cl{M}(q) \TwoVec{\baseVel}{\dot{q}},
\end{equation}
where $\map{\cl{M}(q)}{\spBaseVel\times T_q Q_m}{\spBaseVel^*\times T_q^* Q_m}$ denotes the mass matrix of the VMS given by
\begin{equation}
    \cl{M}(q) = \TwoTwoMat{M_b(q)}{M_{bm}(q)}{M_{bm}^\top(q)}{M_m(q)},
\end{equation}
with
\begin{align}
    M_b(q) & = (\cl{S}_b^{b,s})^\top \left[ \cl{I}_b^b + \sum_{i=1}^n \text{Ad}_{H_b^i(q)}^\top \cl{I}_i^i \text{Ad}_{H_b^i(q)} \right] \cl{S}_b^{b,s}, \label{eq:mass_matrix_Mb} \\
    M_{bm}(q) & = \sum_{i=1}^n (\cl{S}_b^{b,s})^\top \text{Ad}_{H_b^i(q)}^\top \cl{I}_i^i J_i^{i,b}(q), \label{eq:mass_matrix_Mbm}\\
    M_m(q) & = \sum_{i=1}^n (J_i^{i,b}(q))^\top \cl{I}_i^i J_i^{i,b}(q), \label{eq:mass_matrix_Mm}
\end{align}
referred to as the locked, coupling, and manipulator inertia matrices, respectively.

The equations of motion of the VMS can then be derived using the recursive Newton-Euler algorithm \cite{lynch2017modern} and expressed in the Lagrangian form as
\begin{equation}\label{eq:floating_base_Lagrangian_dynamics}
    \cl{M}(q) \TwoVec{\dot{\baseVel}}{\ddot{q}} + \cl{C}(\baseVel,q,\dot{q}) \TwoVec{\baseVel}{\dot{q}} = \TwoVec{\baseWrench}{\tau},
\end{equation}
where $\map{\cl{C}(\baseVel,q,\dot{q}) }{\spBaseVel\times T_q Q_m}{\spBaseVel^*\times T_q^* Q_m}$ contains the Coriolis and centrifugal terms,
$\baseWrench \in \spBaseVel^*$ denotes the external wrench acting on the moving base expressed in frame $\{b\}$, and $\tau \in T_q^* Q_m$ denotes the joint generalized forces.
The explicit expression of the Coriolis matrix $\cl{C}(\baseVel,q,\dot{q})$ can be found in \cite{Mishra2023}.

The above equations of motion lack structure that is essential for model-based control design and analysis. For instance, it is not even straightforward to
show, without tedious factorization of $\cl{C}$, conservation of energy or passivity properties from \eqref{eq:floating_base_Lagrangian_dynamics} due to the lack of this structure.
The main contribution of this paper is to reveal this structure by formulating the dynamics of VMS in the port-Hamiltonian framework.

\begin{table}
    \centering
    \begin{tabular}{|c|c|c|}
        \hline
        \textbf{Map/Variable} & \textbf{Coordinate-free} & \textbf{Coordinate-based} \\
        \hline
        $(\baseVel,\dot{q})$ & $\spBaseVel\times T_qQ_m$ & $\R{b} \times \R{n}$ \\
        $(\baseWrench,\tau)$ & $\spBaseVel^*\times T_q^* Q_m$ & $\R{b} \times \R{n}$ \\
        $J_i^{i,s}(q)$ & $\spBaseVel\times T_q Q_m \rightarrow \mathfrak{g}_i^s$ & $\R{6\times(b+n)}$ \\
        $J_i^{i,b}(q)$ & $T_q Q_m \rightarrow \mathfrak{g}_i^b$ & $\R{6\times n}$ \\
        $M_b(q)$ & $\spBaseVel\rightarrow{\spBaseVel^*}$ & $\R{b\times b}$ \\
        $M_{bm}(q)$ & $T_q Q_m \rightarrow{\spBaseVel^*}$ & $\R{b\times n}$ \\
        $M_m(q)$ & $T_q Q_m \rightarrow{T_q^* Q_m}$ & $\R{n\times n}$\\
        $\cl{M}(q), \cl{C}(\baseVel,q,\dot{q})$ & ${\spBaseVel\times T_q Q_m}\rightarrow{\spBaseVel^*\times T_q^* Q_m}$ & $\R{(b+n)\times(b+n)}$ \\
        \hline
    \end{tabular}
    \caption{Coordinate-free and coordinate-based representations of maps and variables in \eqref{eq:floating_base_Lagrangian_dynamics}}
    \label{table:matrix_representations}
\end{table}

\begin{remark}
    Note that working directly with the pair $(\baseVel, \dot{q}) \in \mathfrak{g}_b \times T_q Q_m$ allows us to avoid the use of local coordinates for representing the
    configuration $(h,q)\in Q$ of the VMS.
    This singularity-free representation is particularly useful when dealing with large rotational motions of the floating base, 
    as it avoids the ambiguities associated with local coordinate representations such as Euler angles.
\end{remark}

\begin{remark}
    We have opted to present the geometric formulation using coordinate-free notation to highlight the underlying geometric structure of each map.
    However, it is more common in the robotics literature to present these maps using matrix and vector representations in a specific coordinate system.
    Therefore, we provide the representations for maps and variables introduced in this section in Table \ref{table:matrix_representations} and we shall do the same
    in the subsequent sections whenever applicable.
\end{remark}

\section{Manipulator Port-Hamiltonian Dynamics} \label{sec:manip_dyn}
In this section we present the \pH formulation of a fixed-base manipulator dynamics including the effects of gravity, actuator torques, 
and interaction wrenches acting on the end-effector.
This section and the next serve as a starting point to ease the subsequent derivation of the \pH dynamics of VMS in Sec. \ref{sec:fm_dyn} as 
well as to contrast the \pH formulation of VMS with that of the base or manipulator alone.

From the geometric mechanics perspective, the starting point is the kinetic energy of an $n$-DoF manipulator which defines 
a Lagrangian system on the tangent bundle $TQ_m$ of the $n$-dimensional configuration manifold $Q_m$.
The Lagrangian function $\map{\subTxt{\cl{L}}{kin}}{TQ_m}{\bb{R}}$ is given by
\begin{equation}\label{eq:manipulator_Lagrangian}
    \subTxt{\cl{L}}{kin}(q,\dot{q}) = \half \dot{q}^\top M_m(q) \dot{q},  
\end{equation}
where the inertia matrix $M_m(q)$ is given by \eqref{eq:mass_matrix_Mm}.

The standard equations of motion of the manipulator are given by the Euler-Lagrange equations corresponding to the Lagrangian $\subTxt{\cl{L}}{kin}$.
Alternatively on the cotangent bundle $T^*Q_m \cong \R{2n}$, the manipulator dynamics can be equivalently represented in \pH form as \cite{duindam2009modeling_a,duindam2009modeling_b}
\begin{align} 
\TwoVec{\dot{q}}{\dot{\pi}} &= J_{\text{sym}} \TwoVec{\parGrad{\subTxt{\cl{H}}{kin}}{q}}{\parGrad{\subTxt{\cl{H}}{kin}}{\pi}} + G \tau, \nonumber\\ 
\dot{q} &= G^\top \TwoVec{\parGrad{\subTxt{\cl{H}}{kin}}{q}}{\parGrad{\subTxt{\cl{H}}{kin}}{\pi}}, \label{eq:manipulator_pH_dynamics}
\end{align}  
where $\tau \in T_q^* Q_m$ denotes external torques acting on the manipulator joints, $\pi \in T^*_q Q_m \cong \R{n}$ denotes the generalized momenta conjugate to $\dot{q}$, defined as 
\begin{equation}\label{eq:manipulator_momentum_def}
    \pi :=\parGrad{\subTxt{\cl{L}}{kin}}{\dot{q}}(q,\dot{q})= M_m(q) \dot{q},
\end{equation}
while $G$ represents the input matrix, and $J_{\text{sym}} = - J_{\text{sym}}^\top$ characterizes the canonical symplectic structure on $T^*Q_m$ given by
$$G:= \TwoVec{0_{n\times n}}{\bb{I}_n},  \qquad J_{\text{sym}} := \TwoTwoMat{0_{n\times n}}{\bb{I}_n}{-\bb{I}_n}{0_{n\times n}}.$$
Moreover, $\map{\subTxt{\cl{H}}{kin}}{T^*Q_m}{\bb{R}}$ denotes the kinetic Hamiltonian function of the manipulator given by
\begin{equation}\label{eq:manipulator_Hamiltonian}
    \subTxt{\cl{H}}{kin}(q,\pi) = \half \pi^\top M_m^{-1}(q) \pi,
\end{equation}
derived from the Legendre transform of the Lagrangian $\subTxt{\cl{L}}{kin}$.
It follows that the partial gradients of the Hamiltonian $\subTxt{\cl{H}}{kin}$ with respect to $q$ and $\pi$, respectively are given by
\begin{align}
\parGrad{\subTxt{\cl{H}}{kin}}{q}(q,\pi) &= - C_m^\top(q,\dot{q}) \dot{q} \in T_q^*Q_m\\
\parGrad{\subTxt{\cl{H}}{kin}}{\pi}(q,\pi) &= \dot{q}\in T_q Q_m
\end{align}
which are expressed as functions of $q$ and $\pi$ through the relation $\dot{q} = M_m^{-1}(q) \pi$,
while $\map{C_m(q,\dot{q})}{T_q Q_m}{T_q^* Q_m}$ denotes the canonical Coriolis-centrifugal matrix of the manipulator satisfying $\dot{M}_m (q) = C_m(q,\dot{q}) + C_m^\top(q,\dot{q})$.\\



\begin{remark}
    We shall sometimes omit the functional dependencies of gradients of functions on their variables for the sake of readability, e.g., writing $\parGrad{\subTxt{\cl{H}}{kin}}{q}$ instead of
    $\parGrad{\subTxt{\cl{H}}{kin}}{q}(q,\pi)$ as in \eqref{eq:manipulator_pH_dynamics} when the context is clear.\\
\end{remark}

\begin{figure}
    \centering
    \includegraphics[width=0.5\columnwidth]{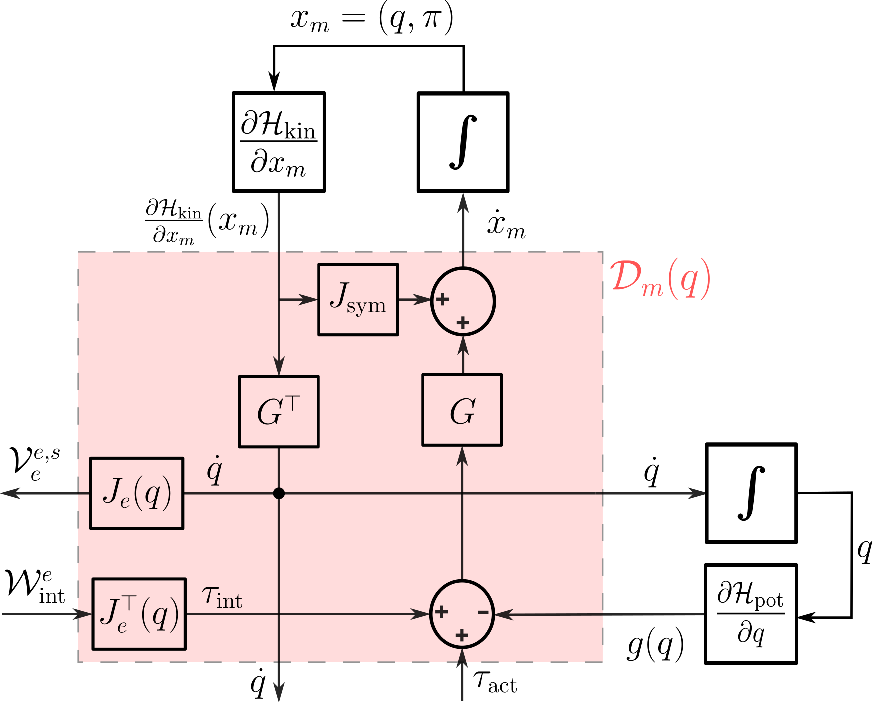}
    \caption{Port-Hamiltonian formulation of fixed-base manipulator dynamics.}
    \label{fig:pH_formulation_manipulator}
\end{figure}

The conservation of kinetic energy is straightforward to verify as follows.
Let $\map{\pair{\cdot}{\cdot}}{\bb{V^*}\times \bb{V}}{\bb{R}}$ denote the dual pairing between elements of any vector space $\bb{V}$ and its dual $\bb{V}^*$.
If we denote by $x_m := (q,\pi) \in T^*Q_m$ the state of the manipulator, the rate of change of the Hamiltonian \eqref{eq:manipulator_Hamiltonian} along trajectories 
of the system dynamics \eqref{eq:manipulator_pH_dynamics} can be expressed as
\begin{equation}\label{eq:manipulator_power_balance}
    \subTxt{\dot{\cl{H}}}{kin} = \pair{\parGrad{\subTxt{\cl{H}}{kin}}{x_m}}{\dot{x}_m}
    = \pair{\parGrad{\subTxt{\cl{H}}{kin}}{x_m}}{J_{\text{sym}} \parGrad{\subTxt{\cl{H}}{kin}}{x_m} + G \tau}
    =\pair{\parGrad{\subTxt{\cl{H}}{kin}}{\pi}}{\tau} = \pair{\dot{q}}{\tau} = \dot{q}^\top \tau,
\end{equation} 
which follows directly from the skew-symmetry of $J_{\text{sym}}$.


The pair $(\tau,\dot{q})$ defines a power port which can change the kinetic energy of the system.
A typical manipulator's kinetic energy would change due to the power supplied by the actuators, power due to work done by wrenches acting on the end effector, 
and conservative torques due to potential energy fields such as gravity. Consequently, we have that
\begin{equation} \label{eq:manipulator_total_power}
    \pair{\tau}{\dot{q}} = \pair{\subTxt{\tau}{act}}{\dot{q}} + \pair{\subTxt{\cl{W}}{int}^e}{\twist{e}{s}{e}} - \subTxt{\dot{\cl{H}}}{pot}(q),
\end{equation}
where $\subTxt{\tau}{act} \in T_q^* Q_m$ denotes the actuator torques, $\subTxt{\cl{W}}{int}^e \in (\R{6})^*$ denotes the wrench acting on the end-effector, $\twist{e}{s}{e} \in \R{6}$ denotes the 
end-effector's body twist defined by
\begin{equation} \label{eq:end_effector_twist}
    \twist{e}{s}{e} = J_e(q) \dot{q},
\end{equation}
with $J_e(q) \equiv J_e^{e,s}(q)$ denoting the manipulator Jacobian of the end-effector,
and $\subTxt{\cl{H}}{pot}: Q_m \rightarrow \bb{R}$ denotes the potential energy function of the manipulator.

\begin{proposition}
The torque $\tau \in T_q^* Q_m$ in the port-Hamiltonian dynamics \eqref{eq:manipulator_pH_dynamics} that characterizes the power balance \eqref{eq:manipulator_total_power} is given by
\begin{equation} \label{eq:manipulator_total_torque}
    \tau = \subTxt{\tau}{act} + J_e^\top(q)\subTxt{\cl{W}}{int}^e - g(q),
\end{equation}
where $g(q) := \parGrad{\subTxt{\cl{H}}{pot}}{q}(q) \in T_q^* Q_m$ denotes the generalized gravitational torque acting on the manipulator.
Furthermore, the port-Hamiltonian dynamics \eqref{eq:manipulator_pH_dynamics} and \eqref{eq:manipulator_total_torque} is equivalent to the standard manipulator equations of motion given by
\begin{equation}\label{eq:manipulator_EL_eqns_with_torque}
    M_m(q) \ddot{q} + C_m(q,\dot{q}) \dot{q} + g(q) = \subTxt{\tau}{act} + J_e^\top(q) \subTxt{\cl{W}}{int}^e.
\end{equation}
\end{proposition}

\begin{proof}
The expression for the torque $\tau$ in \eqref{eq:manipulator_total_torque} follows directly from substituting \eqref{eq:end_effector_twist} into \eqref{eq:manipulator_total_power} and rearranging terms as:
\begin{align*}
    \pair{\tau}{\dot{q}} &= \pair{\subTxt{\tau}{act}}{\dot{q}} + \pair{\subTxt{\cl{W}}{int}^e}{J_e(q) \dot{q}} - \pair{\parGrad{\subTxt{\cl{H}}{pot}}{q}}{\dot{q}}, \\
    &= \pair{\subTxt{\tau}{act}}{\dot{q}} + \pair{J^{\top}_e(q) \subTxt{\cl{W}}{int}^e}{\dot{q}} - \pair{g(q)}{\dot{q}}
    = \pair{\subTxt{\tau}{act} + J^{\top}_e(q) \subTxt{\cl{W}}{int}^e - g(q)}{\dot{q}},
\end{align*}
which holds for all $\dot{q} \in T_q Q_m$.
To show the equivalence between the port-Hamiltonian dynamics \eqref{eq:manipulator_pH_dynamics} and \eqref{eq:manipulator_total_torque} and the standard manipulator equations of motion \eqref{eq:manipulator_EL_eqns_with_torque},
we substitute \eqref{eq:manipulator_total_torque} into \eqref{eq:manipulator_pH_dynamics} and use the relations for the partial gradients of the Hamiltonian $\subTxt{\cl{H}}{kin}$
with respect to $q$ and $\pi$ to obtain
\begin{equation*}
    \dot{\pi} = - C_m^\top(q,\dot{q}) \dot{q} + \subTxt{\tau}{act} + J^{\top}_e(q) \subTxt{\cl{W}}{int}^e - g(q).
\end{equation*} 
Using the relation $\dot{\pi} = M_m(q) \ddot{q} + \dot{M}_m(q) \dot{q}$ and rearranging terms then yields the standard manipulator equations of motion \eqref{eq:manipulator_EL_eqns_with_torque}.
\end{proof}

We conclude this section with a number of remarks that highlight the importance of the above constructions and their relevance to the subsequent developments in this paper.
\begin{enumerate}

    \item The \pH model provides a physically meaningful representation of the manipulator dynamics by highlighting power conjugate variables that define so called power ports. 
    Each power port consists of a pair of variables whose product has the physical dimension of mechanical power, as depicted in Fig. \ref{fig:pH_formulation_manipulator}.
    In particular, $(\parGrad{\subTxt{\cl{H}}{kin}}{x_m}, \dot{x}_m) \in T_{x_m}^* (T^* Q_m) \times T_{x_m} (T^* Q_m)$ characterizes the rate of change of the kinetic energy of the manipulator, 
    $(\parGrad{\subTxt{\cl{H}}{pot}}{q}, \dot{q}) \in T_q^* Q_m \times T_q Q_m$ characterizes the rate of change of the potential energy of the manipulator, $(\subTxt{\cl{W}}{int}^e,\twist{e}{s}{e}) \in (\R{6})^*\times \R{6}$ 
    characterizes the power supplied from the environment to the end-effector,and $(\subTxt{\tau}{act}, \dot{q}) \in T_q^* Q_m \times T_q Q_m$ characterizes the power supplied to the manipulator through the actuators.

    \item Combining \eqref{eq:manipulator_power_balance} and \eqref{eq:manipulator_total_power} then yields
    \begin{equation} \label{eq:manipulator_power_balance_total}
    \subTxt{\dot{\cl{H}}}{kin}(q,\pi) + \subTxt{\dot{\cl{H}}}{pot}(q) = \pair{\subTxt{\tau}{act}}{\dot{q}} + \pair{\subTxt{\cl{W}}{int}^e}{\twist{e}{s}{e}},
    \end{equation}
    which expresses the power balance of the manipulator system, stating that the rate of change of the total energy (kinetic + potential) of the manipulator is equal to the total power supplied to the manipulator
    through the actuators and interaction.
    In fact, \eqref{eq:manipulator_power_balance_total} also highlights the fact that the manipulator system is passive with respect to the input-output pair $(\subTxt{\tau}{act}, \dot{q})$ and $(\subTxt{\cl{W}}{int}^e, \twist{e}{s}{e})$.
    Such property is crucial for designing stable energy-based controllers for manipulators \cite{ortega2002putting,van2000l2}.

    \item The \pH model \eqref{eq:manipulator_pH_dynamics} and \eqref{eq:manipulator_total_torque} extends the standard Hamiltonian formulation 
    to include external power ports that allow for energy exchange with the environment.
    This extended power balance \eqref{eq:manipulator_power_balance_total} is characterized by the Dirac structure $\cl{D}_m(q)$ identified by
    \begin{equation*}
        \begin{pmatrix}
        \dot{x}_m \\
        \dot{q} \\
        \dot{q} \\
        \twist{e}{s}{e}
        \end{pmatrix}
        = 
        \begin{pmatrix}
        \subTxt{J}{sym} & -G    & G     & G J_e^\top(q) \\
        G^\top  & \cdot & \cdot & \cdot  \\
        G^\top  & \cdot & \cdot & \cdot  \\
        J_e(q) G^\top & \cdot & \cdot & \cdot
        \end{pmatrix}
        \begin{pmatrix}
        \parGrad{\subTxt{\cl{H}}{kin}}{x_m} \\
        \parGrad{\subTxt{\cl{H}}{pot}}{q} \\
        \subTxt{\tau}{act} \\
        \subTxt{\cl{W}}{int}^e
        \end{pmatrix}
    \end{equation*}
    where the dots represent zero blocks of appropriate dimensions, omitted for brevity.

    The composition of Dirac structures through power ports preserves the Dirac structure property \cite{van2000l2}, which is a powerful property of \pH systems that allows for 
    modular modeling of complex systems and systematic design of energy-based controllers \cite{Rashad2022,stramigioli2015energy}.

\end{enumerate}

\section{Base Port-Hamiltonian Dynamics} \label{sec:base_dyn}
Following the same line of thought in the previous section, we now derive the \pH formulation of the moving base dynamics (without a manipulator) treated as a single rigid body.
The kinetic energy of the base defines a Lagrangian given by
\begin{equation}
    \subTxt{\cl{L}}{kin}(\baseVel) = \frac{1}{2} \baseVel^\top \baseIner \baseVel,
\end{equation}
where $\map{\baseIner}{\spBaseVel}{\spBaseVel^*}$ is the projected inertia tensor of the base defined by $\baseIner:= (\cl{S}_b^{b,s})^\top \cl{I}^b_b \cl{S}_b^{b,s} $.
The above Lagrangian defines a Lagrangian system on the tangent bundle $TG_b$ of the base configuration manifold $G_b$ using the map $\chi_h$ in \eqref{eq:chi_map_def}
relating $\baseVel \in \spBaseVel$ to $(h,\dot{h})\in T G_b$.

Using the Legendre transformation, we define the Hamiltonian of the base as
\begin{equation}\label{eq:base_Hamiltonian}
    \subTxt{\cl{H}}{kin}(\baseMom) = \frac{1}{2} \baseMom^\top \baseIner^{-1} \baseMom,
\end{equation}
where the generalized momentum of the base $\baseMom \in \spBaseVel^* $ is given by 
\begin{equation}\label{eq:base_momentum_def}
    \baseMom :=\parGrad{\subTxt{\cl{L}}{kin}}{\baseVel}(\baseVel) = \baseIner \baseVel.
\end{equation}
The above Hamiltonian can also be extended to a function on the cotangent bundle $T^*G_b$ using the dual map $\map{\chi_h^*}{\spBaseVel^*}{T_h^* G_b}$.

\begin{figure}
    \centering
    \includegraphics[width=0.5\columnwidth]{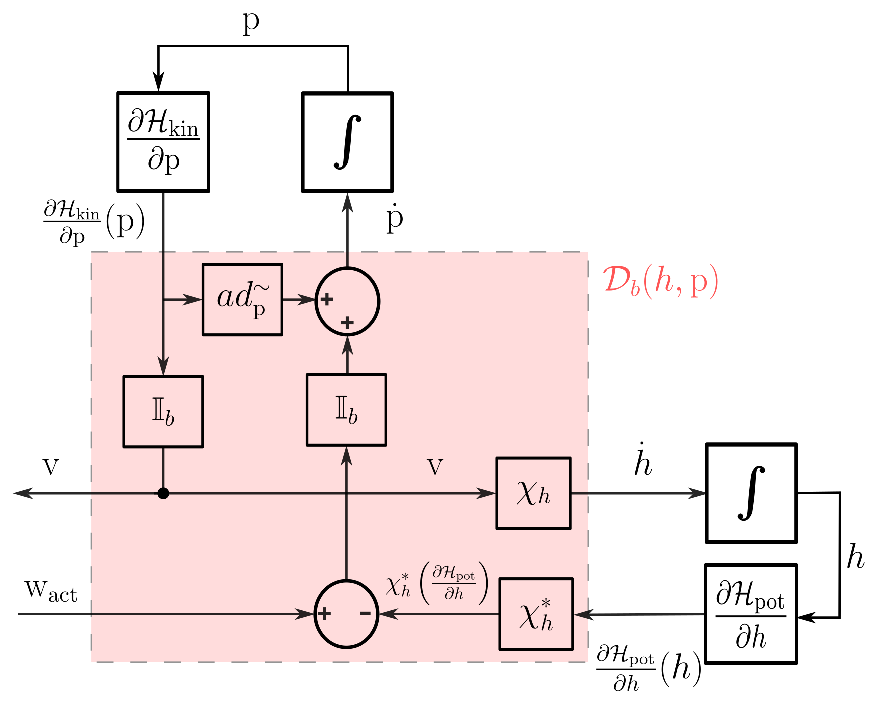}
    \caption{Port-Hamiltonian formulation of moving-base dynamics.}
    \label{fig:pH_formulation_base}
\end{figure}

Due to the invariance of the Hamiltonian \eqref{eq:base_Hamiltonian} under the left action of $G_b$ on $T^*G_b$, we can reduce the canonical 
Hamiltonian equations to the dual Lie algebra $\spBaseVel^*$.
The reduced Hamiltonian equations of motion of the base are given by the Lie-Poisson equations \cite{Marsden1993}, which can be extended using D'Alembert principle 
to include external wrenches acting on the base as \cite{rashad2021energy}:
\begin{align} 
    \dot{\baseMom} &= ad_\baseMom^\sim \parGrad{\subTxt{\cl{H}}{kin}}{\baseMom} + \baseWrench, \nonumber\\
    \baseVel &= \parGrad{\subTxt{\cl{H}}{kin}}{\baseMom}, \label{eq:base_pH}
\end{align}
where $\baseWrench \in \spBaseVel^*$ denotes the external wrench acting on the base, $ad_{\baseMom}^\sim : \spBaseVel \rightarrow \spBaseVel^*$ is the 
skew-symmetric operator defined as
\begin{equation}\label{eq:ad_sim_def}
    ad_\baseMom^\sim \bar{\baseVel} = ad_{\bar\baseVel}^\top \baseMom, \quad \forall \bar{\baseVel} \in \spBaseVel,
\end{equation}
that represents the Lie-Poisson structure on $\spBaseVel^*$, and $\parGrad{\subTxt{\cl{H}}{kin}}{\baseMom} (\baseMom) = \baseIner^{-1} \baseMom = \baseVel \in \spBaseVel$.
The power balance of the above \pH system is given by
\begin{equation}\label{eq:base_power_balance}
    \subTxt{\dot{\cl{H}}}{kin}(\baseMom) = \pair{\parGrad{\subTxt{\cl{H}}{kin}}{\baseMom}}{\dot{\baseMom}} = \pair{\baseVel}{\baseWrench},
\end{equation}
which follows from the skew-symmetry of $ad_\baseMom^\sim$.
The external wrench $\baseWrench$ can be used to model power supplied due to actuators or gravitational potential energy such that
\begin{equation}\label{eq:base_added_power}
    \pair{\baseWrench }{\baseVel}= \pair{\subTxt{\baseWrench}{act} }{\baseVel} - \subTxt{\dot{\cl{H}}}{pot}(h).
\end{equation}



\begin{proposition}
The base wrench $\baseWrench \in \spBaseVel^*$ in the port-Hamiltonian dynamics \eqref{eq:base_pH} that characterizes the power balance \eqref{eq:base_added_power} is given by
\begin{equation} \label{eq:base_total_wrench}
    \baseWrench = \subTxt{\baseWrench}{act} - \chi_h^* \left(\parGrad{\subTxt{\cl{H}}{pot}}{h}\right).
\end{equation}
\end{proposition}
\begin{proof}
The proof follows by subsituting $\subTxt{\dot{\cl{H}}}{pot}(h) = \pair{\parGrad{\subTxt{\cl{H}}{pot}}{h}}{\dot{h}}$ and  $\dot{h} = \chi_h(\baseVel)$ into \eqref{eq:base_added_power} and rearranging terms as:
\begin{equation*}
    \pair{\baseWrench}{\baseVel} = \pair{\subTxt{\baseWrench}{act}}{\baseVel} - \pair{\parGrad{\subTxt{\cl{H}}{pot}}{h}}{\dot{h}}
    = \pair{\subTxt{\baseWrench}{act}}{\baseVel} - \pair{\chi_h^* \left(\parGrad{\subTxt{\cl{H}}{pot}}{h}\right)}{\baseVel}
    = \pair{\subTxt{\baseWrench}{act} - \chi_h^* \left(\parGrad{\subTxt{\cl{H}}{pot}}{h}\right)}{\baseVel},
\end{equation*}
which holds for all $\baseVel \in \spBaseVel$.
\end{proof}

The total power balance of the \pH system \eqref{eq:base_pH} with the base wrench \eqref{eq:base_total_wrench} is given by
\begin{equation}\label{eq:base_total_power}
    \subTxt{\dot{\cl{H}}}{kin}(\baseMom) + \subTxt{\dot{\cl{H}}}{pot}(h) = \pair{\subTxt{\baseWrench}{act}}{\baseVel},
\end{equation}
stating that the rate of change of the total energy (kinetic + potential) of the base is equal to the total power supplied to it through its actuators.
By assuming the potential energy of the base is bounded from below, the above power balance implies that the base \pH system \eqref{eq:base_pH} 
is passive with respect to the input-output pair $(\subTxt{\baseWrench}{act},\baseVel)$.

The Dirac structure $\cl{D}_b(h,\baseMom)$ associated with the base \pH system \eqref{eq:base_pH} and \eqref{eq:base_total_wrench} 
characterizing the power balance \eqref{eq:base_total_power} is identified by
\begin{equation}
    \ThrVec{\dot{\baseMom}}{\dot{h}}{\baseVel} = 
    \begin{pmatrix}
        ad_{ \baseMom }^\sim & -\chi_h^*  & \bb{I}_b \\
        \chi_h & \cdot & \cdot \\
        \bb{I}_b & \cdot & \cdot
    \end{pmatrix}
    \ThrVec{\parGrad{\subTxt{\cl{H}}{kin}}{\baseMom}}{\parGrad{\subTxt{\cl{H}}{pot}}{h}}{\subTxt{\baseWrench}{act}},
\end{equation}
and graphically represented in Fig. \ref{fig:pH_formulation_base}.



\section{Vehicle-Manipulator Port-Hamiltonian Dynamics} \label{sec:fm_dyn}
With the above constructions, we can now present the main contributions of this work which is the \pH formulation of the dynamics of VMS.

\subsection{Port-Hamiltonian formulation}
Our starting point is the kinetic energy function of a VMS \eqref{eq:fm_kinetic_energy} which defines a reduced Lagrangian on $\spBaseVel^* \times T^*Q_m \cong TQ$ given by
\begin{equation}\label{eq:Lagrangian_fm}
    \subTxt{\cl{L}}{kin}(\baseVel,q,\dot{q}) = \half \TwoVec{\baseVel}{\dot{q}}^\top \cl{M}(q) \TwoVec{\baseVel}{\dot{q}},
\end{equation}
which leads to the conjugate momentum variables $\baseMom \in \spBaseVel^*$ and $\pi \in T_q^* Q_m$ defined by
\begin{align}
    \baseMom :=& \parGrad{\subTxt{\cl{L}}{kin}}{\baseVel}(\baseVel,q,\dot{q}) = M_{b}(q) \baseVel + M_{bm}(q) \dot{q}, \label{eq:fm_base_momentum_def}\\
    \pi :=& \parGrad{\subTxt{\cl{L}}{kin}}{\dot{q}}(\baseVel,q,\dot{q}) = M_{bm}^\top(q) \baseVel + M_{m}(q) \dot{q}. \label{eq:fm_manipulator_momentum_def}
\end{align}

\begin{remark}
    Note that in contrast to the manipulator momentum \eqref{eq:manipulator_momentum_def} which depends only on the joint velocities $\dot{q}$, 
the conjugate momentum variable $\pi$ for a VMS \eqref{eq:fm_manipulator_momentum_def} depends on both the base velocity $\baseVel$ and joint velocities $\dot{q}$ due to the presence of the coupling inertia matrix $M_{bm}(q)$.
The same applies to \eqref{eq:fm_base_momentum_def} in contrast to the base momentum \eqref{eq:base_momentum_def} of a moving base.
From this point onwards, we will denote by $\baseMom$ and $\pi$ the base and manipulator momenta of a VMS as defined in \eqref{eq:fm_base_momentum_def} and \eqref{eq:fm_manipulator_momentum_def}, respectively.\\
\end{remark}

The \pH dynamics of a VMS that characterize the conservation of kinetic energy in the presence of external wrenches on the base and actuator torques is given by the following theorem.\\

\begin{theorem} \label{theorem:pH_formulation}
    Let $x:=(p,q,\pi) \in \cl{X}$ denote the state of a VMS, with the state space denoted by $\cl{X}:=\spBaseVel^*\times T^*Q_m \cong \R{b+2n}$
    The equations of motion governing the state $x\in \cl{X}$ in \pH form are given by
\begin{align}
        \ThrVec{\dot{\baseMom}}{\dot{q}}{\dot{\pi}} =&
        \cl{J}(p)
        \ThrVec{\parGrad{\subTxt{\cl{H}}{kin}}{\baseMom}}{\parGrad{\subTxt{\cl{H}}{kin}}{q}}{\parGrad{\subTxt{\cl{H}}{kin}}{\pi}} 
        + \cl{G}
        \TwoVec{\baseWrench}{\tau}, \nonumber \\
        \TwoVec{\baseVel}{\dot{q}}=& 
        \cl{G}^\top
        \ThrVec{\parGrad{\subTxt{\cl{H}}{kin}}{\baseMom}}{\parGrad{\subTxt{\cl{H}}{kin}}{q}}{\parGrad{\subTxt{\cl{H}}{kin}}{\pi}} ,\label{eq:floating_base_manipulator_pH}
\end{align}
where 
\begin{equation}
    \cl{J}(\baseMom) := \begin{pmatrix}
            ad_\baseMom^\sim & \cdot & \cdot \\
            \cdot & \cdot & \bb{I}_n \\
            \cdot & -\bb{I}_n & \cdot\\
        \end{pmatrix}, \quad 
        \cl{G} := \begin{pmatrix}
            \bb{I}_b & \cdot \\
            \cdot & \cdot  \\
            \cdot &  \bb{I}_n
        \end{pmatrix}  ,
\end{equation}
denote the interconnection and input matrices, respectively, $\baseWrench \in \spBaseVel^*$ denotes the resultant wrench \eqref{eq:base_total_wrench}
acting on the base, and $\tau \in \bb{R}^n$ denotes the resultant torques \eqref{eq:manipulator_total_torque} applied at the manipulator's joints.

The Hamiltonian function $\map{\subTxt{\cl{H}}{kin}}{\cl{X}}{\bb{R}}$ is given by the Legendre transform of the reduced kinetic energy Lagrangian \eqref{eq:Lagrangian_fm} as
\begin{equation}\label{eq:Hamiltonian_fm}
    \subTxt{\cl{H}}{kin}(\baseMom,q,\pi) = \half \TwoVec{\baseMom}{\pi}^\top \cl{M}^{-1}(q) \TwoVec{\baseMom}{\pi}.
\end{equation}
\end{theorem}
\begin{proof}
See Appendix \ref{sec:proof_of_pH_theorem}.
\end{proof}

\begin{corollary}\label{corollary:pH_formulation}
    The Hamiltonian \eqref{eq:Hamiltonian_fm} satisfies along trajectories $(\baseMom(t), q(t), \pi(t)) \in \cl{X}$ of the \pH dynamics \eqref{eq:floating_base_manipulator_pH} the power balance
\begin{equation}\label{eq:fm_power_balance} 
    \dot{\cl{H}}_{kin} = \pair{\baseWrench}{\baseVel} + \pair{\tau}{\dot{q}},
\end{equation}
stating that the rate of change of the total kinetic energy of the VMS is equal to the total power supplied to it through the base wrench and joint torques.
\end{corollary}
\begin{proof}
See Appendix \ref{sec:proof_of_pH_theorem}.
\end{proof}

\begin{table}
    \centering
    \begin{tabular}{|c|c|c|}
        \hline
        \textbf{Map/Variable} & \textbf{Coordinate-free} & \textbf{Coordinate-based} \\
        \hline
        $(\dot{\baseMom}, \parGrad{\subTxt{\cl{H}}{kin}}{\baseMom})$ & $\spBaseVel^*\times \spBaseVel$ & $\R{b} \times \R{b}$ \\
        $(\dot{q}, \parGrad{\subTxt{\cl{H}}{kin}}{q})$ & $T_q Q_m \times T_q^* Q_m$ & $\R{n} \times \R{n}$ \\
        $(\dot{\pi}, \parGrad{\subTxt{\cl{H}}{kin}}{\pi})$ & $T_q^* Q_m \times T_q Q_m$ & $\R{n} \times \R{n}$ \\
        $ad_\baseMom^\sim$ & $\spBaseVel\rightarrow\spBaseVel^*$ & $\R{b\times b}$ \\
        $\cl{J}(\baseMom)$ & $T_x^*\cl{X} \rightarrow T_x\cl{X}$ & $\R{(b+2n)\times(b+2n)}$ \\
        $\cl{G}$ & $\spBaseVel^*\times T_q^* Q_m \rightarrow T_x\cl{X}$ & $\R{(b+2n)\times(b+n)}$ \\
        \hline
    \end{tabular}
    \caption{Coordinate-free and coordinate-based representations of maps and variables in \eqref{eq:floating_base_manipulator_pH}}
    \label{table:matrix_representations_pH_VMS}
\end{table}

The interconnection and input matrices of the \pH dynamics \eqref{eq:floating_base_manipulator_pH}, derived from first principles using reduction, can be seen to be a 
direct combination of those of the base \pH dynamics \eqref{eq:base_pH} and manipulator \pH dynamics \eqref{eq:manipulator_pH_dynamics}. This is a direct consequence of the fact
that the Dirac structure of the coupled system is the direct product of the Dirac structures of the individual subsystems and is independent from the Hamiltonian function.
The only source of coupling between the base and manipulator dynamics in \eqref{eq:floating_base_manipulator_pH} is through the configuration-dependent inertia matrix $\cl{M}(q)$ defining
the Hamiltonian \eqref{eq:Hamiltonian_fm}. 

The \pH dynamics \eqref{eq:floating_base_manipulator_pH} is useful for analysis purposes (e.g., to show passivity) but not suitable for control design or numerical implementation 
due to the coupling between the base and manipulator momenta. 
This source of difficulty is mainly reflected in the Hamiltonian gradients appearing in \eqref{eq:floating_base_manipulator_pH} which can be expressed in terms of the velocities $\baseVel$ and $\dot{q}$ as
\begin{equation*}
    \parGrad{\subTxt{\cl{H}}{kin}}{\baseMom}(\baseVel,q,\dot{q}) = \baseVel, \qquad   
    \parGrad{\subTxt{\cl{H}}{kin}}{\pi}(\baseVel,q,\dot{q}) = \dot{q}, \qquad  
    \parGrad{\subTxt{\cl{H}}{kin}}{q}(\baseVel,q,\dot{q}) = -\half \parGrad{}{q} \left[\baseVel^\top M_b(q) \baseVel + 2 \baseVel^\top M_{bm}(q) \dot{q} + \dot{q}^\top M_m(q) \dot{q} \right].
\end{equation*}
Writing the gradients in terms of the momenta $\baseMom$ and $\pi$ requires the computation of the inverse of the inertia matrix $\cl{M}^{-1}(q)$ 
with respect to $q$ which is cumbersome and computationally expensive due to the coupling inertia matrix $M_{bm}(q)$.

In principle, one should be able to show that the \pH dynamics \eqref{eq:floating_base_manipulator_pH} is equivalent to the standard equations of motion of VMS given in \eqref{eq:floating_base_Lagrangian_dynamics}
by performing the inverse Legendre transform and substituting the definitions of the momenta. However, this process is quite tedious again due to $M_{bm}(q)$ and we shall not pursue it here.
Instead, we will present in the coming subsection an alternative \pH formulation that is more suitable for control design and simulation and 
show its equivalence to the reduced Euler-Lagrange equations and other formulations in the literature.

\subsection{Inertially-decoupled port-Hamiltonian formulation}
The configuration space $Q=G_b\times Q_m$ of a VMS has the special structure of a trivial principal bundle with base manifold $Q_m$ and 
structure group $G_b$, as depicted in Fig. \ref{fig:principal_bundle_structure}.
This special structure allows us to define a new set of momentum variables that decouple the base and manipulator dynamics inertially (i.e., at the acceleration level).
Specifially, we can decompose a tangent vector $(\dot{h},\dot{q}) \in T_{(h,q)}Q$, or equivalently $(\baseVel,\dot{q})\in \spBaseVel\times T_qQ_m$, into vertical motions
along the fibers (locked manipulator with $\dot{q}=0$) and 
horizontal motions defined such that
\begin{equation}
    \cl{A}_b(h,\dot{h},q,\dot{q}) = 0,
\end{equation} 
where $\map{\cl{A}_b}{TQ}{\spBaseVel}$ is the natural mechanical connection on the principal bundle $Q$ defined by \cite{Marsden1993}
\begin{equation}
    \cl{A}_b(h,\dot{h},q,\dot{q}) := \baseVel + A(q) \dot{q},
\end{equation}
where $\baseVel$ is related to $(h,\dot{h})$ via $\chi_h$ and $\map{A(q)}{T_q Q_m}{\spBaseVel}$ is the local connection form defined by
\begin{equation}
    A(q) := M_b^{-1}(q) M_{bm}(q).
\end{equation}

\begin{figure}
    \centering
    \includegraphics[width=0.5\columnwidth]{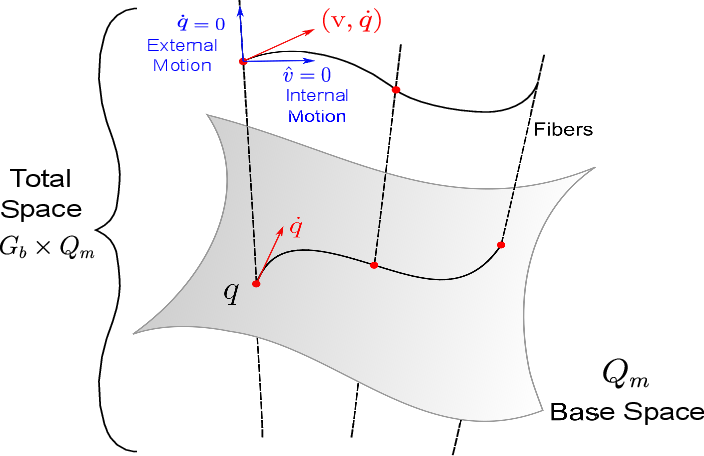}
    \caption{Principal bundle structure of the VMS configuration space $Q = G_b \times Q_m$.}
    \label{fig:principal_bundle_structure}
\end{figure}

Thus, the principal bundle structure of $Q$ provides a natural way to decompose the motion of a VMS into internal motions, defined by the manipulator joints $\dot{q} \in T_qQ_m$, 
and external motions, defined by the so called locked velocity (twist) of the base given by
\begin{equation}
    \hat\baseVel := \baseVel + A(q) \dot{q} \in \spBaseVel.
\end{equation}
It is straightforward to see that the Lagrangian \eqref{eq:Lagrangian_fm} can be rewritten in terms of $(\hat\baseVel, q, \dot{q})\in \spBaseVel\times TQ_m$ as \cite{Moghaddam2024}
\begin{equation}\label{eq:decoupled_Lagrangian_fm}
    \subTxt{\hat{\cl{L}}}{kin}(\hat\baseVel,q,\dot{q}) = \half \hat\baseVel^\top M_b(q) \hat\baseVel + \half \dot{q}^\top \hat{M}_m(q) \dot{q},
\end{equation}
where $\hat{M}_m(q) := M_m(q) - A^\top(q) M_{bm}(q)$ denotes the Schur complement of $M_b(q)$ in $\cl{M}(q)$.

In what follows, we shall change the coordinates of the \pH dynamics to be represented in the conjugate momenta associated with the Lagrangian \eqref{eq:decoupled_Lagrangian_fm} with respect to
the variables $(\hat\baseVel, q, \dot{q})$.
For that, we shall need the following two lemmas.\\

\begin{lemma}\label{lemma:phi_map}
Let $\map{\phi(q)}{\cl{X}}{\cl{X}}$ denote the diffeomorphism mapping $x:=(\baseMom,q,\pi)$ to $\hat{x}:=(\hat\baseMom,\hat q,\hat{\pi})$ defined by
\begin{equation}\label{eq:diffeomorphism_phi}
    \hat\baseMom := \baseMom, \qquad \hat q = q, \qquad \hat{\pi} := - A^\top(q) \baseMom + \pi,
\end{equation}
where $\hat \baseMom := \parGrad{\subTxt{\hat{\cl{L}}}{kin}}{\hat \baseVel} \in \spBaseVel^*$ and $\hat\pi := \parGrad{\subTxt{\hat{\cl{L}}}{kin}}{\dot{\hat q}} \in T_q^* Q_m$ are the conjugate momenta associated 
with the Lagrangian \eqref{eq:decoupled_Lagrangian_fm} with respect to the variables $\hat\baseVel$ and $ \dot{\hat q}$, respectively.
\end{lemma}
\begin{proof}
The proof follows directly by substituting the definitions of $\baseMom$ and $\pi$ in terms of $(\baseVel,q,\dot{q})$ into the definitions of $\hat\baseMom$ and $\hat{\pi}$.
For brevity, we shall drop the dependence of the maps on $q$ in the rest of the proof.
We have
\begin{align*}
    \hat\baseMom = \parGrad{\subTxt{\hat{\cl{L}}}{kin}}{\hat \baseVel} = M_b \hat\baseVel = M_b (\baseVel + A \dot{q}) = M_b \baseVel + M_b M_b^{-1} M_{bm} \dot{q} = \baseMom,
\end{align*}
\begin{align*}
    \hat{\pi} = \parGrad{\subTxt{\hat{\cl{L}}}{kin}}{\dot{q}} = \hat{M}_m \dot{q} = M_m \dot{q} - M_{bm}^\top M_b^{-1} M_{bm} \dot{q} 
    =& M_m \dot{q} - A^\top M_{bm} \dot{q}  = M_m \dot{q} - A^\top ( \baseMom - M_b \baseVel) \\
    =&  M_m \dot{q} + M_{bm}^\top \baseVel - A^\top \baseMom = \pi - A^\top \baseMom,
\end{align*}
which proves the lemma.
\end{proof}

\begin{table}
    \centering
    \begin{tabular}{|c|c|c|}
        \hline
        \textbf{Map} & \textbf{Coordinate-free} & \textbf{Coordinate-based} \\
        \hline
        $A(q)$ & ${T_{q} Q_m}\rightarrow{\spBaseVel}$ & $\R{b\times n}$ \\
        $\Phi(\baseMom,q)$ & $T_x\cl{X} \rightarrow T_{\hat{x}}\cl{X}$ & $\R{(b+2n)\times(b+2n)}$ \\
        $L(\baseMom,q)$ & $T_q Q_m \rightarrow T_q^*Q_m$ & $\R{n\times n}$ \\
        $N_b(\hat{\baseMom},\hat{q})$ & $T_{\hat{q}} Q_m \rightarrow \spBaseVel$ & $\R{b\times n}$ \\
        $\hat{N}_m(\hat{q},\hat{\pi})$ & $T_{\hat{q}} Q_m \rightarrow T_{\hat{q}}^* Q_m$ & $\R{n\times n}$ \\
        $B(\hat\baseMom,\hat{q})$ & ${T_{\hat{q}} Q_m}\rightarrow{T_{\hat{q}}^* Q_m}$ & $\R{n\times n}$ \\
        $\hat{\cl{J}}(\hat\baseMom,\hat{q})$ & $T_x^*\cl{X} \rightarrow T_x\cl{X}$ & $\R{(b+2n)\times(b+2n)}$ \\
        $\hat{\cl{G}}(\hat{q})$ & $\spBaseVel^*\times T_q^* Q_m \rightarrow T_x\cl{X}$ & $\R{(b+2n)\times(b+n)}$ \\
        \hline
    \end{tabular}
    \caption{Coordinate-free and coordinate-based representations of maps in \eqref{eq:decoupled_floating_base_manipulator_pH}}
    \label{table:matrix_representations_decoupled_pH_VMS}
\end{table}


\begin{lemma} \label{lemma:tangent_phi}
    The tangent map of the diffeomorphism $\phi(q)$ at a point $x$, denoted by $\Phi(\baseMom,q) := T_{x} \phi(q)$, is given by the map $\map{\Phi(\baseMom,q)}{T_x \cl{X}}{T_{\hat{x}}\cl{X}}$ expressed as
    \begin{equation}
        \Phi(\baseMom,q) = \begin{pmatrix}
            \bb{I}_b & \cdot & \cdot \\
            \cdot & \bb{I}_n & \cdot \\
            - A^\top(q)& - L(\baseMom,q) & \bb{I}_n
        \end{pmatrix},
    \end{equation}
    where $\map{L(\baseMom,q)}{T_q Q_m}{T_q^*Q_m}$ is defined by
    \begin{equation}\label{eq:L_map_def}
        L(\baseMom,q) := \begin{bmatrix}
             \parGrad{A^\top}{q_1 }(q) \baseMom,  \cdots,
             \parGrad{A^\top}{q_2}(q) \baseMom
        \end{bmatrix}.
    \end{equation}
\end{lemma}
\begin{proof}
Let $(\delta \baseMom, \delta q, \delta \pi) \in T_x \cl{X}$ be arbitrary variations at $x=(\baseMom,q,\pi) \in \cl{X}$.
The tangent map $\Phi(q)$ is defined by the relation
\begin{equation*}
    \ThrVec{\delta \hat\baseMom}{\delta \hat{q}}{\delta \hat{\pi}} = \Phi(\baseMom,q) \ThrVec{\delta \baseMom}{\delta q}{\delta \pi},
\end{equation*}
where $(\delta \hat\baseMom, \delta \hat{q}, \delta \hat{\pi}) \in T_{\hat{x}} \cl{X}$ are the variations induced by $(\delta \baseMom, \delta q, \delta \pi)$ through the diffeomorphism $\phi(q)$.
The first two rows of the tangent map $\Phi(\baseMom,q)$ follow directly from the first two equalities in \eqref{eq:diffeomorphism_phi}.
To derive the third row, we compute the differential of the third equality in \eqref{eq:diffeomorphism_phi} as follows:
\begin{equation*}
    \delta\hat{\pi} = \delta\pi - \delta(A^\top(q) \baseMom)=\delta\pi - A^\top(q) \delta\baseMom - \delta A^\top(q) \baseMom,
\end{equation*}
Using $\delta A^\top(q) = \sum_{i=1}^{n} \parGrad{A^\top}{q_i}(q) \delta q_i$, we can write
\begin{equation*}
    \delta A^\top(q) \baseMom =  L(\baseMom,q) \delta q,
\end{equation*}
which proves the lemma.
\end{proof}

\begin{figure*}
    \centering
    \includegraphics[width=\textwidth]{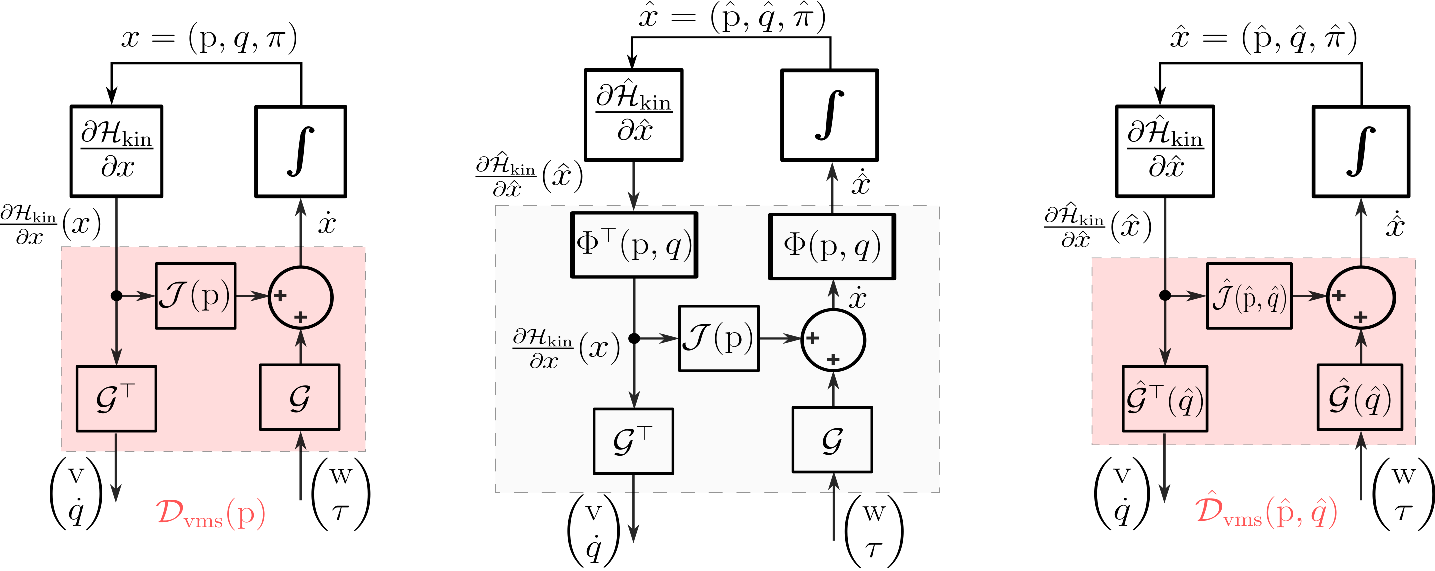}
    \caption{The proposed port-Hamiltonian formulations of VMS dynamics and their corresponding Dirac structures. Left figure shows the \pH formulation \eqref{eq:floating_base_manipulator_pH} while the right 
    figure shows the inertially-decoupled \pH formulation \eqref{eq:decoupled_floating_base_manipulator_pH}. The intermediate figure illustrates the change of coordinates $\phi(q)$.}
    \label{fig:pH_formulation_VMS}
\end{figure*}

\begin{theorem}\label{theorem:pH_formulation_decoupled}
    The inertially-decoupled equations of motion governing the state $\hat{x}:=(\hat\baseMom,\hat{q},\hat{\pi}) \in \cl{X}$ of a VMS in \pH form are given by
\begin{align}
        \ThrVec{\dot{\hat\baseMom}}{\dot{\hat{q}}}{\dot{\hat{\pi}}} =&
        \hat{\cl{J}}(\hat\baseMom,\hat{q})
        \ThrVec{\parGrad{\subTxt{\hat{\cl{H}}}{kin}}{\hat\baseMom}}{\parGrad{\subTxt{\hat{\cl{H}}}{kin}}{\hat{q}}}{\parGrad{\subTxt{\hat{\cl{H}}}{kin}}{\hat{\pi}}} 
        + \hat{\cl{G}}(\hat{q})
        \TwoVec{\baseWrench}{\tau} \nonumber \\
        \TwoVec{\baseVel}{\dot{q}}=& 
        \hat{\cl{G}}^\top(\hat{q})
        \ThrVec{\parGrad{\subTxt{\hat{\cl{H}}}{kin}}{\hat\baseMom}}{\parGrad{\subTxt{\hat{\cl{H}}}{kin}}{\hat{q}}}{\parGrad{\subTxt{\hat{\cl{H}}}{kin}}{\hat{\pi}}} ,\label{eq:decoupled_floating_base_manipulator_pH}
\end{align}
where the Hamiltonian function $\map{\subTxt{\hat{\cl{H}}}{kin}}{\cl{X}}{\bb{R}}$ is defined by $\subTxt{\hat{\cl{H}}}{kin} := \subTxt{\cl{H}}{kin} \circ \phi^{-1}(q)$ and expressed as
\begin{equation}\label{eq:Hamiltonian_fm_decoupled}
    \subTxt{\hat{\cl{H}}}{kin}(\hat\baseMom,\hat{q},\hat{\pi}) = \half \hat\baseMom^\top M_b^{-1}(\hat{q}) \hat\baseMom + \half \hat{\pi}^\top \hat{M}_m^{-1}(\hat{q}) \hat{\pi},
\end{equation}
and the gradients are expressed as
\begin{align}
    \parGrad{\subTxt{\hat{\cl{H}}}{kin}}{\hat\baseMom}(\hat\baseMom,\hat{q},\hat{\pi}) =& M_b^{-1}(\hat{q}) \hat\baseMom , \label{eq:Hamiltonian_grad_p_hat} \\
    \parGrad{\subTxt{\hat{\cl{H}}}{kin}}{\hat{q}}(\hat\baseMom,\hat{q},\hat{\pi}) =& N_b^\top(\hat{\baseMom},\hat{q}) \hat{\baseMom} + \hat{N}_m^\top(\hat{q},\hat{\pi}) \hat{\pi} \label{eq:Hamiltonian_grad_q_hat} \\
    \parGrad{\subTxt{\hat{\cl{H}}}{kin}}{\hat{\pi}}(\hat\baseMom,\hat{q},\hat{\pi}) =& \hat{M}_m^{-1}(\hat{q}) \hat{\pi}, \label{eq:Hamiltonian_grad_pi_hat} 
\end{align}
where $\map{N_b(\hat{\baseMom},\hat{q})}{T_{\hat{q}} Q_m}{\spBaseVel}$ and $\map{\hat{N}_m(\hat{q},\hat{\pi})}{T_{\hat{q}} Q_m}{T_{\hat{q}}^* Q_m}$ are defined by
\begin{align}\label{eq:N_b_and_N_m_definitions}
    N_b(\hat{\baseMom},\hat{q}) :=& \half \left[ \parGrad{M_b^{-1}}{\hat{q}_1}(\hat{q}) \hat{\baseMom},  \cdots, \parGrad{M_b^{-1}}{\hat{q}_n}(\hat{q}) \hat{\baseMom} \right], \\
    \hat{N}_m(\hat{q},\hat{\pi})  :=&\half  \left[ \parGrad{\hat{M}_m^{-1}}{\hat{q}_1}(\hat{q}) \hat{\pi},  \cdots, \parGrad{\hat{M}_m^{-1}}{\hat{q}_n}(\hat{q}) \hat{\pi} \right].
\end{align}

The interconnection and input matrices are defined by the relations
\begin{align}
    \hat{\cl{J}}(\hat\baseMom,\hat{q}) :=& \Phi(\baseMom,q) \cl{J}(\baseMom) \Phi^\top(\baseMom,q), \label{eq:J_fm_hat_definition}\\
    \hat{\cl{G}}(\hat{q}) :=& \Phi(\baseMom,q) \cl{G}. \label{eq:G_fm_hat_definition}
\end{align}
with $q = \hat{q}$ and $\baseMom = \hat\baseMom$, and expressed as
\begin{align}
    \hat{\cl{J}}(\hat\baseMom,\hat{q}) :=& \begin{pmatrix}
            ad_{\hat\baseMom}^\sim & \cdot & - ad_{\hat\baseMom}^\sim A(\hat{q}) \\
            \cdot & \cdot & \bb{I}_n \\
            -A^\top(\hat{q}) ad_{\hat\baseMom}^\sim & -\bb{I}_n & -B(\hat\baseMom,\hat{q})
        \end{pmatrix}, \label{eq:J_fm_hat_expression}\\
         \hat{\cl{G}}(\hat{q}) :=& \begin{pmatrix}
            \bb{I}_b & \cdot \\
            \cdot & \cdot  \\
            -A^\top(\hat{q})  &  \bb{I}_n
        \end{pmatrix}, \label{eq:G_fm_hat_expression}
\end{align}
where $\map{B(\hat\baseMom,\hat{q})}{T_{\hat{q}} Q_m}{T_{\hat{q}}^* Q_m}$ is the skew-symmetric map defined by
\begin{equation}\label{eq:B_definition}
    B(\hat\baseMom,\hat{q}) := - A^\top(\hat{q}) ad_{\hat\baseMom}^\sim A(\hat{q}) + L(\hat\baseMom,\hat{q}) - L^\top(\hat\baseMom,\hat{q}).
\end{equation}
\end{theorem}

\begin{proof}
See Appendix \ref{sec:proof_of_pH_theorem_decoupled}.
\end{proof}



\begin{corollary}\label{corollary:pH_formulation_decoupled}
The Hamiltonian \eqref{eq:Hamiltonian_fm_decoupled} satisfies along trajectories $(\hat\baseMom(t), \hat q(t), \hat \pi(t)) \in \cl{X}$ of the \pH dynamics 
\eqref{eq:decoupled_floating_base_manipulator_pH} the power balance
\begin{equation}\label{eq:fm_power_balance_decoupled_1} 
    \subTxt{\dot{\hat{\cl{H}}}}{kin} = \pair{\baseWrench}{\baseVel} + \pair{\tau}{\dot{q}},
\end{equation}
stating that the rate of change of the total kinetic energy of the VMS is equal to the total power supplied due to the resultant base wrench and joint torques.
The power balance can also be expressed equivalently as
\begin{equation}\label{eq:fm_power_balance_decoupled_2} 
    \subTxt{\dot{\hat{\cl{H}}}}{kin} = \pair{\hat{\baseWrench}}{\hat{\baseVel}} + \langle \hat{\tau} | \dot{\hat{q}} \rangle ,
\end{equation}
where $\hat{\baseWrench} \in \spBaseVel^*$ and $\hat{\tau} \in T_q^*Q_m$ are defined by
\begin{equation}\label{eq:hat_inputs}
    \hat{\baseWrench} := \baseWrench, \qquad \qquad \hat{\tau} := \tau - A^\top (\hat{q})\baseWrench ,
\end{equation}
and represent the power conjugate variables to the locked velocity $\hat{\baseVel}$ and joint velocities $\dot{\hat{q}}$, respectively.

\end{corollary}
\begin{proof}
See Appendix \ref{sec:proof_of_pH_theorem_decoupled}.
\end{proof}

\section{Lagrangian Counterparts of VMS Dynamics} \label{sec:fm_dyn_lagrangian}
In this section we show the equivalence between the \pH formulation of the inertially-decoupled VMS dynamics \eqref{eq:decoupled_floating_base_manipulator_pH}
and its Lagrangian counterparts given in the works of \cite{Marsden1993,Muller2023,Mishra2023,Moghaddam2024}.

We shall start by the following Lemmas that establish a number of properties of the mass matrices $M_b$ and $\hat{M}_m$ as well as the relation between the gradients of the Hamiltonian \eqref{eq:Hamiltonian_fm_decoupled} 
and the velocity variables $(\hat\baseVel, q, \dot{q})$ used in Lagrangian formulations.

\begin{lemma}\label{lemma:mass_matrix_properties}
The mass matrices $M_b$ and $\hat{M}_m$ defined in (\ref{eq:mass_matrix_Mb},\ref{eq:mass_matrix_Mm}) satisfy the following properties:
\begin{align}
    \dot{M}_m(q) \dot{q} &= \hat{C}_m(q,\dot{q}) \dot{q} + \hat{C}_m^\top(q,\dot{q}) \dot{q} \label{eq:mass_matrix_property_Mm} \\
    \parGrad{}{q} \left( \dot{q}^\top \hat{M}_m(q) \dot{q} \right) &= 2\ \hat{C}_m^\top(q,\dot{q}) \dot{q} \label{eq:mass_matrix_property_grad_Mm}\\
    \dot{M}_b(q) \hat{\baseVel} &= E_b(\hat{\baseVel},q) \dot{q} + P_b(q,\dot{q}) \hat{\baseVel} \label{eq:mass_matrix_property_Mb}\\
    \parGrad{}{q} \left( \hat{\baseVel}^\top M_b(q) \hat{\baseVel} \right) &= 2\ E_b^\top(\hat{\baseVel},q) \hat{\baseVel} \label{eq:mass_matrix_property_grad_Mb}
\end{align}
where $\map{\hat{C}_m(q,\dot{q})}{T_q Q_m}{T_q^* Q_m}$ denotes the canonical Coriolis matrix associated to $M_b$,
while $\map{E_b(\hat{\baseVel}, q)}{T_q Q_m}{\spBaseVel}$ is defined by
\begin{equation}
    E_b(\hat{\baseVel}, q) := \half \left[\parGrad{M_b}{q_1}(q) \hat\baseVel,  \cdots,  \parGrad{M_b}{q_n}(q) \hat\baseVel \right],
\end{equation}
and $\map{P_b(q,\dot{q})}{\spBaseVel}{\spBaseVel}$ is defined by
\begin{equation}
    P_b(q,\dot{q}) := \half \sum_{k=1}^n \parGrad{M_b}{q^k}(q) \dot{q}^k.
\end{equation}
\end{lemma}
\begin{proof}
The proof in this lemma will be presented using index notation for clarity and compactness. Furthermore, repeated indices imply summation over their
range following Einstein's summation convention. The range for indices $i,j,k$ is from $1$ to $n$ while for indices $I,J$ it is from $1$ to $b$.

i) We have that in component form, the time derivative of $\hat{M}_m$ can be expressed as $[\dot{\hat{M}}_m]_{ij} = \partial_k \hat{M}_{ij} \dot{q}^k$,
where $\hat{M}_{ij}$ denote the components of $\hat{M}_m$ and $\partial_k \hat{M}_{ij}$ is short for $\parGrad{\hat{M}_{ij}}{q^k}(q)$.

The components of $\hat{C}_m$ are given by the Christoffel symbols of the first kind as
\begin{equation}
    [\hat{C}_m]_{ij} = \half \left( \partial_k \hat{M}_{ij} + \partial_j \hat{M}_{ik} - \partial_i \hat{M}_{jk} \right) \dot{q}^k.
\end{equation}
It is straightforward to verify that $[\hat{C}_m \dot{q}]_{i} = (\partial_k \hat{M}_{ij} - \half \partial_i \hat{M}_{jk})\dot{q}^j\dot{q}^k $ and that
$[\hat{C}_m^\top \dot{q}]_{i} = \half \partial_i \hat{M}_{jk} \dot{q}^j\dot{q}^k$. Therefore, adding these two expressions yields \eqref{eq:mass_matrix_property_Mm}.

ii) The gradient of the quadratic form $\dot{q}^\top \hat{M}_m(q) \dot{q}$ with respect to $q$ can be computed as
\begin{equation}
    \parGrad{}{q^i} \left( \dot{q}^\top \hat{M}_m(q) \dot{q} \right) = \partial_i \hat{M}_{jk} \dot{q}^j \dot{q}^k,
\end{equation}
which is clearly equal to $2\ [\hat{C}_m^\top \dot{q}]_{i}$

iii) In component form, the time derivative of $M_b$ can be expressed as $[\dot{M}_b]_{IJ} = \partial_k M_{IJ} \dot{q}^k$,
where $M_{IJ}$ denote the components of the matrix $M_b$. The components of $E_b$ are defined as $[E_b]_{Ik} = \half \partial_k M_{IJ} \hat{\baseVel}^J$.
It is straightforward to verify that $[E_b \dot{q}]_{I} = \half \partial_k M_{IJ} \hat{\baseVel}^J \dot{q}^k$ and that
$[P_b \hat{\baseVel}]_{I} = \half \partial_k M_{IJ} \dot{q}^k \hat{\baseVel}^J$. Therefore, adding these two expressions yields \eqref{eq:mass_matrix_property_Mb}.

iv) Finally similar to ii), the gradient of the quadratic form $\hat{\baseVel}^\top M_b(q) \hat{\baseVel}$ with respect to $q$ has components $\partial_k M_{IJ} \hat{\baseVel}^I \hat{\baseVel}^J$
which is equal to $2\ [E_b^\top \hat{\baseVel}]_{k}$.
\end{proof}

\begin{lemma}\label{lemma:gradients_relation}
The gradients of the Hamiltonian \eqref{eq:Hamiltonian_fm_decoupled} 
given in (\ref{eq:Hamiltonian_grad_p_hat}-\ref{eq:Hamiltonian_grad_pi_hat}) are given in terms of $(\hat\baseVel, q, \dot{q}) \in \spBaseVel\times TQ_m$ as
\begin{align}
    \parGrad{\subTxt{\hat{\cl{H}}}{kin}}{\hat\baseMom}(\hat\baseVel, q, \dot{q}) =& \hat{\baseVel} , \label{eq:Hamiltonian_grad_p_hat_2} \\
    \parGrad{\subTxt{\hat{\cl{H}}}{kin}}{\hat{q}}(\hat\baseVel, q, \dot{q}) =& - E_b^\top (\hat{\baseVel}, q) \hat{\baseVel} - \hat{C}_m^\top(q,\dot{q}) \dot{q},  \label{eq:Hamiltonian_grad_q_hat_2} \\
    \parGrad{\subTxt{\hat{\cl{H}}}{kin}}{\hat{\pi}}(\hat\baseVel, q, \dot{q}) =& \dot{q}. \label{eq:Hamiltonian_grad_pi_hat_2} 
\end{align}
\end{lemma}
\begin{proof}
The proof of \eqref{eq:Hamiltonian_grad_p_hat_2} and \eqref{eq:Hamiltonian_grad_pi_hat_2} follow directly from $\hat\baseVel = M_b^{-1}(q) \hat\baseMom$ and $\dot{q} = \hat{M}_m^{-1}(q) \hat{\pi}$.
As for the first term in \eqref{eq:Hamiltonian_grad_q_hat}, using the symmetry of $M_b$ and the property $\partial_i M_b^{-1} = - M_b^{-1} \partial_i M_b M_b^{-1} $ 
we have that
$$2[N_b^\top \hat{\baseMom}]_k =  \hat{\baseMom}^\top \partial_k M_b^{-1} \hat{\baseMom} =  \hat{\baseVel}^\top M_b\partial_k M_b^{-1} M_b\hat{\baseVel} 
= -  \hat{\baseVel}^\top \partial_k M_b \hat{\baseVel},$$
which in index notation reads as $2[N_b^\top \hat{\baseMom}]_k = - \partial_k M_{IJ} \hat{\baseVel}^I \hat{\baseVel}^J$.
Similar to the proof of Lemma \ref{lemma:mass_matrix_properties} (iv), we have that $N_b^\top \hat{\baseMom} = - E_b^\top \hat{\baseVel}$.

The second term in \eqref{eq:Hamiltonian_grad_q_hat} can be shown in the same manner as a result of $\hat{N}_m^\top \hat{\pi} = - \hat{C}_m^\top \dot{q}$.
\end{proof}




\begin{table}
    \centering
    \begin{tabular}{|c|c|c|}
        \hline
        \textbf{Map} & \textbf{Coordinate-free} & \textbf{Coordinate-based} \\
        \hline
        $P_b(q,\dot{q})$ & ${\spBaseVel}\rightarrow{\spBaseVel}$ & $\R{b\times b}$ \\
        $E_b(\hat\baseVel,q)$ & ${T_q Q_m}\rightarrow{\spBaseVel}$ & $\R{b\times n}$ \\
        $\hat{C}_m(q,\dot{q})$ & ${T_q Q_m}\rightarrow{T_q^* Q_m}$ & $\R{n\times n}$ \\
        $\hat{B}(\hat\baseVel,\hat{q})$ & ${T_{\hat{q}} Q_m}\rightarrow{T_{\hat{q}}^* Q_m}$ & $\R{n\times n}$ \\
        $\hat{\cl{M}}, \hat{\cl{C}}_1, \hat{\cl{C}}_2$ & ${\spBaseVel \times T_q Q_m}\rightarrow{\spBaseVel^* \times T_q^* Q_m}$ & $\R{(b+n)\times(b+n)}$ \\
        \hline
    \end{tabular}
    \caption{Coordinate-free and coordinate-based representations of maps in \eqref{eq:lagrangian_floating_base_manipulator}}
    \label{table:matrix_representations_Lagrangian_VMS}
\end{table}

\subsection{Equivalence with Reduced Euler-Lagrange equations of \cite{Mishra2023,Moghaddam2024}}

\begin{proposition}\label{proposition:reduced_euler_lagrange_equivalence}
    The \pH dynamics of the VMS given in \eqref{eq:floating_base_manipulator_pH} are equivalent
     to the reduced Euler-Lagrange equations of motion of a VMS given by
\begin{equation}\label{eq:lagrangian_floating_base_manipulator}
    \begin{split}
        \hat{\cl{M}}(q) \TwoVec{\dot{\hat{\baseVel}}}{\ddot{q}} + \hat{\cl{C}}_1(\hat{\baseVel}, q,\dot{q}) \TwoVec{\hat{\baseVel}}{\dot{q}} 
         + \hat{\cl{C}}_2(\hat{\baseVel},q)\TwoVec{\hat{\baseVel}}{\dot{q}}  = \TwoVec{\hat{\baseWrench}}{\hat{\tau}}
    \end{split}
\end{equation}
where 
\begin{align}
    \hat{\cl{M}}(q) :=& \TwoTwoMat{M_b(q)}{0_{b\times n}}{0_{n\times b}}{\hat{M}_m(q)},\label{eq:mass_matrix_decoupled} \\
    \hat{\cl{C}}_1(\hat{\baseVel}, q,\dot{q}) :=& \TwoTwoMat{P_b(q,\dot{q})}{E_b(\hat{\baseVel},q)}{- E_b^\top(\hat{\baseVel},q)}{\hat{C}_m(q,\dot{q})},\label{eq:C1_hat}\\
    \hat{\cl{C}}_2(\hat{\baseVel},q):=& \TwoTwoMat{- ad_{M_b(q) \hat{\baseVel}}^\sim}{ ad_{M_b(q) \hat{\baseVel}}^\sim A(q)}
         {A^\top(q) ad_{M_b(q) \hat{\baseVel}}^\sim }{\hat{B}(\hat\baseVel,q)},
\end{align}
while $\hat{\baseWrench}$ and $\hat{\tau}$ are defined in \eqref{eq:hat_inputs} and $\hat{B}$ is defined from \eqref{eq:B_definition} such that
$$\hat{B}(\hat\baseVel,q) := B(M_b(q) \hat{\baseVel},q).$$
\end{proposition}

\begin{proof}
See Appendix \ref{sec:proof_of_reducedEL_Proposition}.
\end{proof}

\subsection{Equivalence with Boltzmann-Hamel's equations of \cite{Muller2023}}

\begin{proposition}\label{proposition:boltzmann_hamel_equivalence}
    The \pH dynamics of VMS given in \eqref{eq:decoupled_floating_base_manipulator_pH} are equivalent
     to the Boltzmann-Hamel equations given by
\begin{equation}\label{eq:boltzmann_hamel_floating_base_manipulator}
    \begin{split}
        \dt \parGrad{\hat{l}}{\hat{\baseVel}^I} + \left(\gamma_{IJ}^K \hat{\baseVel}^J + \gamma_{Ij}^K \dot{q}^j\right) \parGrad{\hat{l}}{\hat{\baseVel}^K} = \hat{\baseWrench}_I,\\
        \dt \parGrad{\hat{l}}{\dot{q}^i} - \parGrad{\hat{l}}{q^i}  + \left(\gamma_{iJ}^K \hat{\baseVel}^J + \gamma_{ij}^K \dot{q}^j \right)  \parGrad{\hat{l}}{\hat{\baseVel}^K} = \hat{\tau}_i,
    \end{split}
\end{equation}
where the indices $I,J,K$ range from $1$ to $b$ while the indices $i,j$ range from $1$ to $n$.
The Hamel coefficients are defined as
\begin{align}
    \gamma_{IJ}^K :=& c_{IJ}^K, \label{eq:Hamel_coeff_1}\\
    \gamma_{Ij}^K :=& - \gamma_{jI}^K := - c_{IJ}^K A_j^J, \label{eq:Hamel_coeff_2}\\
    \gamma_{ij}^K := & c_{IJ}^K A_i^I A_j^J + \parGrad{A_i^K}{q^j} - \parGrad{A_j^K}{q^i}, \label{eq:Hamel_coeff_3}
\end{align}
while $\map{\hat{l}}{\spBaseVel\times TQ_m}{\bb{R}}$ is given by \eqref{eq:decoupled_Lagrangian_fm}
and $c_{IJ}^K$ are the structure constants of the Lie algebra $\spBaseVel$ defined by $ad_{e_I} e_J = c_{IJ}^K e_K$ for a basis $\{e_I\}$ of $\spBaseVel$.
\end{proposition}
\begin{proof}
    See Appendix \ref{sec:proof_of_boltzmann_hamel_Proposition}.
\end{proof}

\begin{remark}
    Equations \eqref{eq:boltzmann_hamel_floating_base_manipulator} have been referred to the Lagrange-Poincare equations in \cite{Cendra2001,bloch2004nonholonomic} as well
    as the reduced Euler-Lagrange equations in \cite{Marsden1993}.
\end{remark}

\section{Discussion} \label{sec:discussion}
In this section, we highlight the merits of the \pH formulation derived in this paper by comparing it with its Lagrangian counterparts and discussing several directions 
for future research to exploit the proposed framework.

\begin{itemize}

    \item We showed in several steps where does each term in \eqref{eq:lagrangian_floating_base_manipulator} come from in the pH formulation.
    In particular, we have shown which terms come from the symplectic structure of the manipulator, which terms come from the Lie-Poisson structure of the base, 
    and which terms appear due to the tangent and cotangent maps of $\phi(q)$ that represents a change of coordinates to the principal bundle coordinates.
    In the Lagrangian formulation \eqref{eq:lagrangian_floating_base_manipulator}, these terms are all mixed together in the Coriolis matrices $\hat{\cl{C}}_1$ and $\hat{\cl{C}}_2$,
    making it difficult to identify their origins.

    \item We showed that the \pH formulation is efficient in showing energy conservation and passivity properties, thanks to its explicit identification of the interconnection structure
    that simply does not affect any stability or passivity analysis.
    Doing the same using the Lagrangian formulation is more tedious as it uses the relations between $\hat{\cl{M}}$ and $\hat{\cl{C}}$ given in Lemma \ref{lemma:mass_matrix_properties}.
    Performing more advanced control synthesis (e.g. adaptive control, observer-based control, etc.) 
    is also expected to be more straightforward compared to the Lagrangian framework.
    In fact, by using bond graph representations of the \pH system, one can perform passivity analysis simply by visual inspection \cite{Rashad2022}. 

    \item In the derivation of the pH formulation presented in Theorem \ref{theorem:pH_formulation_decoupled}, we have reached it in stages by first deriving the \pH dynamics
    in terms of the original velocity variables $(\baseVel, \dot{q})$ in Theorem \ref{theorem:pH_formulation} and then performing a 
    change of coordinates to the inertially-decoupled velocities $(\hat{\baseVel}, \hat{\dot{q}})$. We have also excluded potential energy from the reduction process 
    by adding it later as an external input to the pH system, allowing us to utilize symmetry and perform Hamiltonian reduction for kinetic energy only.
    This approach contrasts with \cite{Moghaddam2024} where the equations were derived in one step in locked velocity form with respect to variations 
    of $\hat{\baseVel}$.

    \item   The Lagrangian equations \eqref{eq:lagrangian_floating_base_manipulator} correspond to the ones derived in \cite{Mishra2023} for the case of a 
    floating base manipulator.
    A slight difference is that in \cite{Mishra2023} the skew-symmetric part of $\hat{\cl{C}}_1$ involving $E_b$ was included in $\hat{\cl{C}}_2$.
    However, we choose to keep it in $\hat{\cl{C}}_1$ to highlight its origin from the time derivative of mass matrices and to highlight the origin 
    of $\hat{\cl{C}}_2$ from the interconnection matrix \eqref{eq:J_fm_hat_expression}.

    \item In \cite{Moghaddam2024} the same equations \eqref{eq:lagrangian_floating_base_manipulator} were presented. However, the term $E_b^\top$ in \eqref{eq:C1_hat} was missing a minus sign, which lead the authors to formulate the dynamics using a 
    symmetric matrix involving $E_b$ and $E_b^\top$ instead of the skew-symmetric form in \eqref{eq:C1_hat} consistent with \cite{Mishra2023}.
    This in fact can be shown to violate physical laws since a symmetric form would imply that the power associated with these terms is always positive or always negative.
    Thus, it would imply that energy is always injected into or always dissipated from the system by these terms.
    However, an energy analysis was not included in \cite{Moghaddam2024} to reveal this issue.
    This highlights the insight that the \pH framework provides into the energetic structure of the system.
    
\end{itemize}

The proposed \pH framework opens several avenues for future research, including but not limited to:
\begin{itemize}
    \item The passivity property of the \pH dynamics can be exploited to design energy-based controllers for physical interaction using systematic control methods such as 
    energy-balancing passivity-based control \cite{Rashad2019}, interconnection-damping-assignment passivity-based control \cite{yuksel2019aerial}, and control by interconnection \cite{Rashad2022}.
    \item The explicit energetic structure of the \pH formulation can be used to develop energy-preserving numerical integration schemes that yield physically accurate 
    computer simulations of VMS dynamics by preserving energy and momentum \cite{Betsch2007, WANG2026116708}.
    \item In this work we have focused only on unconstrained base movement. The use of Dirac structures in the \pH framework allows the straightforward incorporation of 
    holonomic and non-holonomic constraints \cite{duindam2009modeling_a}, enabling the modeling and control of VMS with legged- and wheeled- locomotion \cite{Hamad2023, KORAYEM20151701}.
\end{itemize}

\section{Conclusion} \label{sec:conclusion}
This paper presents a systematic port-Hamiltonian formulation for vehicle-manipulator systems, providing a unified framework that encompasses 
aerial manipulators, underwater manipulators, space robots, and omnidirectional mobile manipulators.
The key contribution lies in deriving the \pH dynamics from first principles using Hamiltonian reduction theory, 
which explicitly reveals the underlying energetic structure that is obscured in more common Lagrangian formulations.
We have rigorously demonstrated the mathematical equivalence between our \pH framework and existing reduced Euler-Lagrange and Boltzmann-Hamel formulations, 
while highlighting how the \pH structure makes the origins of various terms transparent.
The proposed inertially-decoupled formulation leverages the principal bundle structure of the configuration space, avoiding singularities associated with 
local parameterizations and providing a natural foundation for control synthesis and energy-preserving numerical integration schemes in future work.

\appendix
\section{Global Parametrization of Joints}\label{appendix:joint_subgroups}
\begin{enumerate}
    \item \textbf{1-DoF Revolute Joint:} For the case of a 1-DoF revolute joint, we have that $b=1$, $G_i^j = SO(2)$, $G_i = \bb{S}^1$, and $\mathfrak{g}_i=\bb{R}$.
    The joint configuration $q_i \in \bb{S}^1$ represents the rotation angle of the joint, while $\baseVel_i \in \bb{R}$ represents its angular velocity.
    The map $\chi_{q_i}$ is given by the identity map, i.e. $\dot{q}_i = \chi_{q_i}(\baseVel_i) = \baseVel_i$.
    The relative twist $\twist{i}{j}{i}$ allowed by the joint is given by 
    \begin{equation*}
        \twist{i}{j}{i} = \iota_i(\baseVel_i) =  \cl{S}_i^{i,j} \baseVel_i, \qquad \cl{S}_i^{i,j} := \TwoVec{\hat{n}_i}{q_i \times \hat{n}_i}
    \end{equation*}
    where $\hat{n}_i \in \bb{S}^2$ denotes the unit vector along the axis of rotation of the joint and $q_i \in \R{3}$ denotes 
    a vector from the origin of $\{i\}$ to any point on the joint axis. 
    The joint's generalized force $\baseWrench_i \in \bb{R}$ represents the torque applied by the joint, and is related to the wrench $\cl{W}^i \in (\R{6})^*$ acting 
    on body $i$ by $\baseWrench_i = (\cl{S}_i^{i,j})^\top \cl{W}^i.$

    The map $\varphi_i$ can be expressed in terms of the matrix exponential map \cite{lynch2017modern} as
    \begin{equation}
        H_i^j (q_i) = H_i^j (0) \exp(\tilde{\cl{S}}_i^{i,j} q_i),
    \end{equation}
    where $H_i^j(0) \in SE(3)$ denotes the relative pose of body $i$ with respect to body $j$ when the joint angle is zero.

    \item \textbf{1-DoF Prismatic Joint:} For the case of a 1-DoF prismatic joint, we have that $b=1$, $G_i^j = \bb{R}$, $G_i = \bb{R}$, and $\mathfrak{g}_i=\bb{R}$. 
    The joint configuration $q_i \in \bb{R}$ represents the linear displacement of the joint, $\baseVel_i \in \bb{R}$ represents its linear velocity, and $\baseWrench_i \in \bb{R}$ represents
    the force applied by the joint.
    The same construction as in the revolute joint case holds here, with the only difference that $\cl{S}_i^{i,j} \in \R{6}$ is given by
    \begin{equation*}
        \cl{S}_i^{i,j} := \TwoVec{0_{3\times1}}{\hat{n}_i}
    \end{equation*}
    where $\hat{n}_i \in \bb{S}^2$ denotes the unit vector along the direction of translation of the joint.
    \item \textbf{3-DoF Planar Joint:} A planar joint allows for 2 translational DoFs and 1 rotational DoF about an axis $\hat{n}_i\in\bb{S}^2$ normal to the plane of motion of the joint.
    Thus, we have that $b=3$, $G_i^j = SE(2)$, $G_i = \bb{S}^1\times \R{2}$, and $\mathfrak{g}_i=\bb{R}^3$.
    Assuming for simplicity, that $\hat{n}_i$ is aligned with the $\hat{z}$ axis of the joint frame $\{i\}$.
    The joint configuration is then given by $q_i= (\theta_i,\xi_{x,i},\xi_{y,i})\in \bb{S}^1\times \R{2}$ where $\theta_i \in \bb{S}^1$ denotes the rotation angle of the joint about the axis $\hat{n}_i$ and
    $(\xi_{x,i},\xi_{y,i}) \in \R{2}$ denote the translational components.
    The joint velocity is given by $\baseVel_i = (\omega_i,v_{x,i},v_{y,i}) \in \R{3}$ where $\omega_i \in \R{}$ denotes the angular velocity of the joint and
    $(v_{x,i},v_{y,i}) \in \R{2}$ denote the linear velocity components.
    The map $\chi_{q_i}$ is given by
    \begin{equation}
        \dot{q}_i = \chi_{q_i}(\baseVel_i) = \begin{pmatrix}
        1 & 0 & 0 \\
        0 & c_{\theta_i} & -s_{\theta_i} \\
        0 & s_{\theta_i} & c_{\theta_i}
        \end{pmatrix} \ThrVec{\omega_i}{v_{x,i}}{v_{y,i}},
    \end{equation}
    where $c_{\theta_i} := \cos(\theta_i)$ and $s_{\theta_i} := \sin(\theta_i)$.

    The map $\iota_i$ characterizing the relative twist $\twist{i}{j}{i} \in \R{6}$ allowed by the joint is given by
    \begin{equation*}
        \twist{i}{j}{i} = \iota_i(\baseVel_i) =  \begin{pmatrix}
        0 & 0 & 0 \\
        0 & 0 & 0 \\
        1 & 0 & 0 \\
        0 & 1 & 0 \\
        0 & 0 & 1 \\
        0 & 0 & 0 \\
        \end{pmatrix} \ThrVec{\omega_i}{v_{x,i}}{v_{y,i}}, 
    \end{equation*}
    while the map $\varphi_i$ can be expressed as
    \begin{equation}
        H_i^j (q_i) = 
        \begin{pmatrix}
        c_{\theta_i} & -s_{\theta_i} & 0 & \xi_{x,i} \\
        s_{\theta_i} & c_{\theta_i} & 0 & \xi_{y,i} \\
        0 & 0 & 1 & 0 \\
            0 & 0 & 0 & 1
        \end{pmatrix}
    \end{equation}
    
    \item \textbf{6-DoF Floating Joint:} A floating joint allows for 3 translational DoFs and 3 rotational DoFs.
    Thus, we have that $b=6$, $G_i^j = SE(3)$, $G_i = SE(3)$, and $\mathfrak{g}_i$ is isomorphic to $\R{6}$.
    The map $\varphi_{i}$ and $\iota_i$ are given by the identity maps, i.e. the joint configuration is given by $q_i = H_i^j \in SE(3)$ 
    while the joint velocity is given by $\baseVel_i = \twist{i}{j}{i} \in \R{6}$ and the joint's generalized force is given by $\baseWrench_i = \cl{W}^i \in (\R{6})^*$.
    The map $\chi_{q_i}$ is given explicitly by
    \begin{equation}
        \dot{H}_i^j = (H_i^j)^{-1} \twistTd{i}{j}{i}.
    \end{equation}

\end{enumerate}

\section{Reduced Variational Principle Theorem \cite{Marsden1993}} \label{sec:reduced_variational_Theorem}

Let $\cl{Q}:= G_b \times Q_m$ be a smooth manifold and let $G_b$ be a Lie group acting freely and properly on $\cl{Q}$.
We denote by $\spBaseVel$ the Lie algebra associated with the Lie group $G_b$.
We have that the tangent bundle $T\cl{Q}$ is diffeomorphic to $TG_b\times TQ_m$ which in turn is diffeomorphic to $G_b \times \spBaseVel \times TQ_m$ 
by the left trivialization of the tangent bundle $TG_b$.
Let points of $\cl{Q}$ be denoted by $\mathfrak{q} = (h,q)$ with $h \in G_b$ and $q \in Q_m$.

Let $\mathscr{L}: T\cl{Q} \to \bb{R}$ be a $G_b$-invariant Lagrangian. Then $\mathscr{L}$ satisfies the variational principle
$$\delta \int_a^b \mathscr{L}(\mathfrak{q},\dot{\mathfrak{q}}) \extd t = 0,$$
for variations $\delta \mathfrak{q} \in T_\mathfrak{q}\cl{Q}$ vanishing at the endpoints if and only if the reduced Lagrangian $\map{\cl{L}}{\spBaseVel\times TQ_m}{\bb{R}}$ 
defined by $\cl{L}(\baseVel,q,\dot{q}) := \mathscr{L}(h,\dot{h},q,\dot{q})$ with $\baseVel = \chi_{h^{-1}}(\dot{h}) \in \spBaseVel$ satisfies the reduced variational principle
\begin{equation}\label{eq:reduced_L}
    \delta \int_a^b \cl{L}(\baseVel,q,\dot{q}) \extd t = 0,
\end{equation}
for variations $\delta q \in T_q Q_m$ vanishing at the endpoints and variations $\delta \baseVel \in \spBaseVel$ of the form
$$\delta \baseVel = \dot{\eta} + ad_\baseVel \eta,$$
where $\eta$ is an arbitrary curve in $\spBaseVel$ vanishing at the endpoints
and $\map{ad}{\spBaseVel\times \spBaseVel}{\spBaseVel}$ is the adjoint operator of the Lie algebra $\spBaseVel$.

\section{Derivation of the \pH dynamics \eqref{eq:floating_base_manipulator_pH}} \label{sec:proof_of_pH_theorem}

\subsection{Proof of Theorem \ref{theorem:pH_formulation}}

The derivation of the \pH dynamics \eqref{eq:floating_base_manipulator_pH} of a VMS follows from Appendix \ref{sec:reduced_variational_Theorem} using \eqref{eq:Lagrangian_fm} as the reduced Lagrangian.
The variational principle \eqref{eq:reduced_L} can be rewritten as
\begin{equation}\label{eq:variational_principle_expanded}
    \delta \int_a^b 
    \underbrace{\pair{\parGrad{\subTxt{\cl{L}}{kin}}{\baseVel}}{\delta \baseVel}}_{\text{(i)}} + 
    \underbrace{\pair{\parGrad{\subTxt{\cl{L}}{kin}}{q}}{\delta q} + 
    \pair{\parGrad{\subTxt{\cl{L}}{kin}}{\dot{q}}}{\delta \dot{q}}}_{\text{(ii)}} \extd t = 0,
\end{equation}

By applying the Legendre transformation, we define the generalized momenta $\baseMom := \parGrad{\subTxt{\cl{L}}{kin}}{\baseVel}(\baseVel,q,\dot{q}) \in \spBaseVel^*$
and $\pi := \parGrad{\subTxt{\cl{L}}{kin}}{\dot{q}}(\baseVel,q,\dot{q}) \in T_q^* Q_m$ conjugate to the generalized velocities $\baseVel \in \spBaseVel$ and $\dot{q} \in T_q Q_m$, respectively.
The Hamiltonian $\subTxt{\cl{H}}{kin}: \spBaseVel^* \times T^* Q_m \to \bb{R}$ is then defined as the total energy of the system given by
\begin{equation}\label{eq:Hamiltonian_fm_Definition}
    \subTxt{\cl{H}}{kin}(\baseMom,q,\pi) := \pair{\baseMom}{\baseVel} + \pair{\pi}{\dot{q}} - \subTxt{\cl{L}}{kin}(\baseVel,q,\dot{q}),
\end{equation}
where the generalized velocities $\baseVel$ and $\dot{q}$ are expressed in terms of the generalized momenta $\baseMom$ and $\pi$ via the inverse Legendre transformation.
By construction, it follows that 
\begin{equation}\label{eq:gradients_Hamiltonian_fm}
    \parGrad{\subTxt{\cl{H}}{kin}}{\baseMom} = \baseVel ,\quad
    \parGrad{\subTxt{\cl{H}}{kin}}{\pi} = \dot{q},\quad
    \parGrad{\subTxt{\cl{H}}{kin}}{q} = -\parGrad{\subTxt{\cl{L}}{kin}}{q}.
\end{equation}
It is straightforward to show that \eqref{eq:Hamiltonian_fm_Definition} can be expressed as \eqref{eq:Hamiltonian_fm}.

The term (i) in \eqref{eq:variational_principle_expanded} can be rewritten using the definition of $\baseMom$ and integrated by parts as
\begin{align*}
    \int_a^b \pair{\baseMom}{\delta \baseVel} \extd t 
    = \int_a^b \pair{\baseMom}{\dot{\eta} + ad_\baseVel \eta} \extd t 
    = \int_a^b  -\pair{\dot{\baseMom}}{\eta} + \pair{ad_\baseVel^\top \baseMom}{\eta} \extd t 
    = \int_a^b \pair{-\dot{\baseMom} + ad_\baseMom^\sim \baseVel}{\eta} \extd t,
\end{align*}
where we have used the expression of $\delta \baseVel$ from Appendix \ref{sec:reduced_variational_Theorem}.

Similarly, the term (ii) in \eqref{eq:variational_principle_expanded} can be rewritten using the definition of $\pi$ and integrated by parts as
\begin{align*}
    \int_a^b \pair{\parGrad{\subTxt{\cl{L}}{kin}}{q}}{\delta q} + \pair{\pi}{\delta \dot{q}} \extd t 
    = \int_a^b \pair{\parGrad{\subTxt{\cl{L}}{kin}}{q} - \dot{\pi}}{\delta q} \extd t 
\end{align*}
Combining the two terms and using the definition of the Hamiltonian gradients \eqref{eq:gradients_Hamiltonian_fm}, 
the variational principle \eqref{eq:variational_principle_expanded} becomes
\begin{equation*}
    \int_a^b 
    \pair{-\dot{\baseMom} + ad_\baseMom^\sim \parGrad{\subTxt{\cl{H}}{kin}}{\baseMom}}{\eta} +
    \pair{-\dot{\pi} - \parGrad{\subTxt{\cl{H}}{kin}}{q}}{\delta q} \extd t = 0,
\end{equation*}
for arbitrary variations $\eta$ and $\delta q$ vanishing at the endpoints.
Since $\eta$ and $\delta q$ are arbitrary, the integrand must vanish, which gives the equations of motion
\begin{equation*}
    \dot{\baseMom} = ad_\baseMom^\sim \parGrad{\subTxt{\cl{H}}{kin}}{\baseMom}, \qquad \qquad
    \dot{q} = \parGrad{\subTxt{\cl{H}}{kin}}{\pi}, \qquad \qquad
    \dot{\pi} = - \parGrad{\subTxt{\cl{H}}{kin}}{q}.
\end{equation*}
which can be compactly written as
\begin{equation}\label{eq:Jfm_pH}
    \dot{x} = \cl{J}(x) \parGrad{\subTxt{\cl{H}}{kin}}{x}. 
\end{equation}
which proves $ \cl{J}(x) = \cl{J}(\baseMom)$ in \eqref{eq:floating_base_manipulator_pH}.

To include the external forces and torques acting on the VMS, it directly follows from the Lagrange-d'Alembert principle that the momentum
equations are modified to
\begin{equation*}
    \dot{\baseMom} = ad_\baseMom^\sim \parGrad{\subTxt{\cl{H}}{kin}}{\baseMom} + \baseWrench, \qquad\qquad
    \dot{\pi} = - \parGrad{\subTxt{\cl{H}}{kin}}{q} + \tau.
\end{equation*}

\subsection{Proof of Corollary \ref{corollary:pH_formulation}}
It follows from the skewsymmetry of $ad_\baseMom^\sim$ that $\cl{J}(x)$ is skew-symmetric.
Therefore, we have that along trajectories $x(t) = (\baseMom(t), q(t), \pi(t))$ of the \pH dynamics \eqref{eq:Jfm_pH}, the Hamiltonian satisfies
\begin{equation}\label{eq:power_balance_Jfm}
    \subTxt{\dot{\cl{H}}}{kin} = \pair{\parGrad{\subTxt{\cl{H}}{kin}}{x}}{\dot{x}} = \pair{\parGrad{\subTxt{\cl{H}}{kin}}{x}}{\cl{J}(x) \parGrad{\subTxt{\cl{H}}{kin}}{x}}  = 0.
\end{equation}

Adding the external forces and torques, we have the power balance \eqref{eq:power_balance_Jfm} to be extended as
\begin{align*}
    \subTxt{\dot{\cl{H}}}{kin} = \pair{\parGrad{\subTxt{\cl{H}}{kin}}{x}}{\cl{G}
        \TwoVec{\baseWrench}{\tau}}  = \pair{\cl{G}^\top \parGrad{\subTxt{\cl{H}}{kin}}{x}}{\TwoVec{\baseWrench}{\tau}}
        = \pair{\TwoVec{\baseVel}{\dot{q}}}{\TwoVec{\baseWrench}{\tau}} = \pair{\baseVel}{\baseWrench} + \pair{\dot{q}}{\tau}. 
\end{align*}

\section{Derivation of the inertially decoupled \pH dynamics \eqref{eq:decoupled_floating_base_manipulator_pH}} \label{sec:proof_of_pH_theorem_decoupled}

\subsection{Proof of Theorem \ref{theorem:pH_formulation_decoupled}}

The expression of the Hamiltonian \eqref{eq:Hamiltonian_fm_decoupled} in the inertially-decoupled coordinates follows directly from Lemma \ref{lemma:phi_map} as follows.
Let the map $\map{\Psi(q)}{\spBaseVel\times T_q Q_m}{\spBaseVel\times T_q Q_m}$ be defined as
\begin{equation}\label{eq:psi_map_definition}
    \Psi(q) := \TwoTwoMat{\bb{I}_b}{A(q)}{0_{n\times b}}{\bb{I}_n},
\end{equation}
such that the velocities and corresponding momenta transform as
\begin{equation}\label{eq:transformation_velocities}
\TwoVec{\hat\baseVel}{\dot{\hat{q}}} = \Psi(q) \TwoVec{\baseVel}{\dot{q}}, \qquad \TwoVec{\hat\baseMom}{\hat{\pi}} = \Psi^{-\top}(q) \TwoVec{\baseMom}{\pi}.
\end{equation}
Using \eqref{eq:psi_map_definition}, we can express $\cl{M}(q)$ in \eqref{eq:Lagrangian_fm} as $\cl{M}(q) = \Psi^\top(q) \hat{\cl{M}}(q) \Psi(q)$. Consequently, its inverse is given by
\begin{equation}\label{eq:mass_matrix_inverse_transformation}
    \cl{M}^{-1}(q) = \Psi^{-1}(q) \hat{\cl{M}}^{-1}(q) \Psi^{-\top}(q),
\end{equation}
with $\hat{\cl{M}}(q)$ defined in \eqref{eq:mass_matrix_decoupled}.

Substituting \eqref{eq:transformation_velocities} and \eqref{eq:mass_matrix_inverse_transformation} into \eqref{eq:Hamiltonian_fm}, we obtain the Hamiltonian in the inertially-decoupled coordinates as
\begin{align*}
    \subTxt{\cl{H}}{kin}(\baseMom,q,\pi) =& \half \TwoVec{\baseMom}{\pi}^\top \cl{M}^{-1}(q) \TwoVec{\baseMom}{\pi}
    = \half \TwoVec{\baseMom}{\pi}^\top\Psi^{-1}(q) \hat{\cl{M}}^{-1}(q) \Psi^{-\top}(q) \TwoVec{\baseMom}{\pi}, \\
    =& \half \TwoVec{\hat\baseMom}{\hat{\pi}}^\top \hat{\cl{M}}^{-1}(q) \TwoVec{\hat\baseMom}{\hat{\pi}}
    =  \half \hat\baseMom^\top M_b^{-1}(\hat{q}) \hat\baseMom + \half \hat{\pi}^\top \hat{M}_m^{-1}(\hat{q}) \hat{\pi},
\end{align*}
which proves \eqref{eq:Hamiltonian_fm_decoupled}.

The gradients of the Hamiltonian \eqref{eq:Hamiltonian_grad_p_hat} and \eqref{eq:Hamiltonian_grad_pi_hat} follow directly from its expression. As for the gradient \eqref{eq:Hamiltonian_grad_q_hat}, 
it follows from the definition of the Hamiltonian \eqref{eq:Hamiltonian_fm_decoupled} that
\begin{equation}\label{eq:gradient_H_hat_q}
    \parGrad{\subTxt{\hat{\cl{H}}}{kin}}{\hat{q}} = \half \parGrad{}{\hat{q}} \left(\hat\baseMom^\top M_b^{-1}(\hat{q}) \hat\baseMom + \hat{\pi}^\top \hat{M}_m^{-1}(\hat{q}) \hat{\pi} \right).
\end{equation}
The first term in \eqref{eq:gradient_H_hat_q} can be expressed, using the symmetry of $M_b$, as
$$\parGrad{}{\hat{q}} \left[ \hat\baseMom^\top M_b^{-1}(\hat{q}) \hat\baseMom \right] = \ThrVec{\hat\baseMom^\top \parGrad{M_b^{-1}}{\hat{q}_1}(\hat{q})}{\vdots}{\hat\baseMom^\top \parGrad{M_b^{-1}}{\hat{q}_n}(\hat{q})} \hat\baseMom = N_b^\top(\hat{\baseMom},\hat{q}) \hat{\baseMom},$$ 
with $N_b(\hat{\baseMom},\hat{q})$ defined as in \eqref{eq:N_b_and_N_m_definitions}.
Similarly, the same procedure can be applied to the second term in \eqref{eq:gradient_H_hat_q} to obtain
$$\parGrad{}{\hat{q}} \left[ \hat{\pi}^\top \hat{M}_m^{-1}(\hat{q}) \hat{\pi} \right] = N_m^\top(\hat{\pi},\hat{q}) \hat{\pi}.$$

The derivation of the \pH dynamics \eqref{eq:decoupled_floating_base_manipulator_pH} follows from Lemma \ref{lemma:tangent_phi} 
which establishes the coordinate transformation from $\dot{x} := (\dot{\baseMom},\dot{q},\dot{\pi})$ 
to $\dot{\hat{x}} := (\dot{\hat\baseMom},\dot{\hat{q}},\dot{\hat{\pi}})$
and, by duality, the coordinate transformation of the gradients of any function $\map{\cl{F}}{\cl{X}}{\bb{R}}$ as
\begin{equation}\label{eq:transformation_xdot}
\ThrVec{\dot{\hat\baseMom}}{\dot{\hat{q}}}{\dot{\hat{\pi}}} = \Phi(\baseMom,q) \ThrVec{\dot{\baseMom}}{\dot{q}}{\dot{\pi}}, \qquad  \qquad
\ThrVec{\parGrad{\cl{F}}{\baseMom}}{\parGrad{\cl{F}}{q}}{\parGrad{\cl{F}}{\pi}} = 
\Phi^\top(\baseMom,q) \ThrVec{\parGrad{\hat{\cl{F}}}{\hat\baseMom}}{\parGrad{\hat{\cl{F}}}{\hat{q}}}{\parGrad{\hat{\cl{F}}}{\hat{\pi}}},
\end{equation}
with $\map{\hat{\cl{F}}}{\cl{X}}{\bb{R}}$ defined by $\hat{\cl{F}} := \cl{F} \circ \phi^{-1}(q)$ such that
\begin{align*}
\pair{\parGrad{\cl{F}}{\baseMom}}{\dot{\baseMom}} + \pair{\parGrad{\cl{F}}{q}}{\dot{q}} + \pair{\parGrad{\cl{F}}{\pi}}{\dot{\pi}} =
\pair{\parGrad{\hat{\cl{F}}}{\hat\baseMom}}{\dot{\hat\baseMom}} + \pair{\parGrad{\hat{\cl{F}}}{\hat{q}}}{\dot{\hat{q}}} + \pair{\parGrad{\hat{\cl{F}}}{\hat{\pi}}}{\dot{\hat{\pi}}},
\end{align*}
or equivalently
\begin{equation} \label{eq:F_dot_equality}
    \dot{\cl{F}} = \pair{\parGrad{\cl{F}}{x}}{\dot{x}} = \pair{\parGrad{\hat{\cl{F}}}{\hat{x}}}{\dot{\hat{x}}} = \dot{\hat{\cl{F}}}.
\end{equation}

Now if we combine \eqref{eq:floating_base_manipulator_pH} with \eqref{eq:transformation_xdot}, we obtain
\begin{equation*}
    \dot{\hat{x}} = \Phi(\baseMom,q) \cl{J}(\baseMom) \Phi^\top(\baseMom,q) \parGrad{\subTxt{\hat{\cl{H}}}{kin}}{\hat{x}} + \Phi(\baseMom,q) \cl{G} \TwoVec{\baseWrench}{\tau}.
\end{equation*}
It can be shown by introducing \eqref{eq:J_fm_hat_definition}, \eqref{eq:G_fm_hat_definition} and \eqref{eq:B_definition} that the interconnection and input matrices
in \eqref{eq:J_fm_hat_expression} and \eqref{eq:G_fm_hat_expression} follow by direct computation.

\subsection{Proof of Corollary \ref{corollary:pH_formulation_decoupled}}
The power balance \eqref{eq:fm_power_balance_decoupled_1} of the \pH dynamics \eqref{eq:decoupled_floating_base_manipulator_pH} follows 
directly from \eqref{eq:F_dot_equality} with $\cl{F} = \subTxt{\cl{H}}{kin}$ and $\hat{\cl{F}} = \subTxt{\hat{\cl{H}}}{kin}$, and the skew-symmetry of $\hat{\cl{J}}(\hat x)$.
The equivalent expression \eqref{eq:fm_power_balance_decoupled_2} follows directly from
\begin{align*}
    \subTxt{\dot{\hat{\cl{H}}}}{kin} = \pair{\parGrad{\subTxt{\hat{\cl{H}}}{kin}}{\hat{x}}}{\dot{\hat{x}}} = 
     \pair{\parGrad{\subTxt{\hat{\cl{H}}}{kin}}{\hat{x}}}{\hat{\cl{G}} \TwoVec{{\baseWrench}}{{\tau}}}
     =\pair{\TwoVec{\hat{\baseVel}}{\hat{\dot{q}}}}{\TwoVec{\hat{\baseWrench}}{\hat{\tau}}} = \pair{\hat{\baseVel}}{\hat{\baseWrench}} + \pair{\hat{\dot{q}}}{\hat{\tau}}.
\end{align*}

\section{Proof of Prop. \ref{proposition:reduced_euler_lagrange_equivalence}}\label{sec:proof_of_reducedEL_Proposition}
First, we differentiate the momentum relations $\hat\baseMom = M_b \hat{\baseVel}$ and $\hat{\pi} = \hat{M}_m \dot{q}$ to obtain:
\begin{equation*}
\dot{\hat\baseMom} = \dot{M}_b \hat{\baseVel} + M_b \dot{\hat{\baseVel}}, \qquad
\dot{\hat{\pi}} = \dot{\hat{M}}_m \dot{q} + \hat{M}_m \ddot{q}.
\end{equation*}
Using the mass matrix properties from Lemma \ref{lemma:mass_matrix_properties}, we have:
$\dot{M}_b \hat{\baseVel} = E_b \dot{q} + P_b \hat{\baseVel}$ and
$\dot{\hat{M}}_m \dot{q} = \hat{C}_m \dot{q} + \hat{C}_m^\top \dot{q}$.
Substituting these into the momentum time derivatives:
\begin{align}
\dot{\hat\baseMom} &= E_b \dot{q} + P_b \hat{\baseVel} + M_b \dot{\hat{\baseVel}}, \label{eq:mom_hat_dot_expr}\\
\dot{\hat{\pi}} &= \hat{C}_m \dot{q} + \hat{C}^\top_m \dot{q} + \hat{M}_m \ddot{q}, \label{eq:pi_hat_dot_expr}
\end{align}

Then using the gradient relations from Lemma \ref{lemma:gradients_relation}, we substitute them in the 
first equation of \eqref{eq:decoupled_floating_base_manipulator_pH} to get
\begin{equation}\label{eq:base_mom_dyn_decoupled}
\dot{\hat\baseMom} = ad_{\hat\baseMom}^\sim \hat{\baseVel} - ad_{\hat\baseMom}^\sim A \dot{q} + \hat{\baseWrench},
\end{equation}
which can be rewritten using \eqref{eq:mom_hat_dot_expr} as
\begin{equation}
M_b \dot{\hat{\baseVel}} = -P_b \hat{\baseVel} - E_b \dot{q} + ad_{M_b \hat{\baseVel}}^\sim \hat{\baseVel} - ad_{M_b \hat{\baseVel}}^\sim A \dot{q} + \hat{\baseWrench},
\end{equation}

Similarly, for the third equation in \eqref{eq:decoupled_floating_base_manipulator_pH}, we have that
\begin{equation}\label{eq:man_mom_dyn_decoupled}
\dot{\hat{\pi}} = - A^\top ad_{\hat\baseMom}^\sim \hat{\baseVel} + \hat{C}^\top_m \dot{q} + E_b^\top \hat{\baseVel} - \hat{B} \dot{q} + \hat{\tau},
\end{equation}
which can be rewritten using \eqref{eq:pi_hat_dot_expr} as
\begin{equation}
\hat{M}_m \ddot{q} = -\hat{C}_m \dot{q} - E_b^\top \hat{\baseVel} - A^\top ad_{M_b \hat{\baseVel}}^\sim \hat{\baseVel} - \hat{B} \dot{q} + \hat{\tau}.
\end{equation}
Collecting terms and using the definitions of $\cl{C}_1$ and $\cl{C}_2$ from the proposition concludes the proof.

\section{Proof of Prop. \ref{proposition:boltzmann_hamel_equivalence}}\label{sec:proof_of_boltzmann_hamel_Proposition}
First of all, we start by expressing the momentum dynamics in \eqref{eq:base_mom_dyn_decoupled} and \eqref{eq:man_mom_dyn_decoupled} 
in component form using \eqref{eq:Hamiltonian_grad_q_hat_2} as:
\begin{align}
    \dot{\baseMom}_I =& [ad_\baseMom^\sim]_{IJ} \hat{\baseVel}^J - [ad_\baseMom^\sim]_{IJ} A_j^J \dot{q}^j + \hat{\baseWrench}_I, \label{eq:base_mom_dyn_components}\\
    \dot{\hat{\pi}}_i =& - A_i^I [ad_\baseMom^\sim]_{IJ} \hat{\baseVel}^J - [B(\baseMom,q)]_{ij} \dot{q}^j - \parGrad{\subTxt{\hat{\cl{H}}}{kin}}{q^i} + \hat{\tau}_i, \label{eq:man_mom_dyn_components}
\end{align}
where $A_i^I$ denotes the components of the matrix $A(q)$ and we used the fact that $\hat{\baseMom}=\baseMom$.

From the definition of structural constants $c_{IJ}^K$, it follows that the components of the adjoint operator $ad_\baseVel$ can be expressed as
$$[ad_\baseVel]_J^K = c_{IJ}^K \baseVel^I.$$
Consequently, the components of the map $ad_\baseMom^\sim$, defined in \eqref{eq:ad_sim_def}, can be expressed as
\begin{equation}\label{eq:ad_sim_components}
    [ad_\baseMom^\sim]_{IJ} = - c_{IJ}^K \baseMom_K.    
\end{equation}
Then from the definitions of $L(\baseMom,q)$ in \eqref{eq:L_map_def} and $B(\baseMom,q)$ in \eqref{eq:B_definition}, 
it follows that their components can be expressed using \eqref{eq:ad_sim_components} as
\begin{align*}
    [L(\baseMom,q)]_{ij} =& \parGrad{A_i^I}{q^j} \baseMom_I, \\
    [B(\baseMom,q)]_{ij} =&  \left(c_{IJ}^K A_i^I A_j^J + \parGrad{A_i^K}{q^j} - \parGrad{A_j^K}{q^i}\right) \baseMom_K.
\end{align*}

Using the momenta definitions in Lemma \ref{lemma:phi_map}, we can replace ${\baseMom}_I$ and $\hat{\pi}_i$ in \eqref{eq:base_mom_dyn_components} 
and \eqref{eq:man_mom_dyn_components} using $\parGrad{\hat{l}}{\hat \baseVel^I}$ and $\parGrad{\hat{l}}{\dot{\hat q}^i}$, respectively, 
as well as $\parGrad{\subTxt{\hat{\cl{H}}}{kin}}{q^i}$ using  $-\parGrad{\hat{l}}{q^i}$.
Finally, using the skew-symmetry of the structural constants $c_{IJ}^K = - c_{JI}^K$ and by introducing (\ref{eq:Hamel_coeff_1}-\ref{eq:Hamel_coeff_3}),
we can rearrange the terms to obtain the Boltzmann-Hamel equations \eqref{eq:boltzmann_hamel_floating_base_manipulator}.

\bibliographystyle{elsarticle-num}  
\bibliography{pH_VMS_References.bib}

\end{document}